\documentclass[11pt,twocolumn,twoside]{IEEEtran}
\usepackage[utf8x]{inputenc}
\usepackage{scalefnt}
\usepackage{stackengine}

\usepackage{bbm}
\usepackage{xcolor}
\usepackage{cite}
\usepackage{algorithm}
\usepackage{algpseudocode}
\usepackage{amsmath}

\usepackage{footnote}
\makesavenoteenv{tabular}
\makesavenoteenv{table}
\usepackage{listings}
\usepackage{amsmath}
\usepackage{amsthm,amssymb,graphicx,multirow,amsmath,color,amsfonts}
\usepackage{dsfont}
\allowdisplaybreaks 
\usepackage{color} 
\definecolor{mygreen}{RGB}{28,172,0} 
\usepackage{tikz}
\usepackage{lipsum,adjustbox}
\usetikzlibrary{calc,positioning}
\usepackage{pst-tree,array} 
\definecolor{mylilas}{RGB}{170,55,241}
\usepackage{graphicx}
 \usepackage{multirow}
\usepackage{enumitem}
\usepackage{booktabs}
\newtheorem{lemma}{Lemma}
\newtheorem{theorem}{Theorem}
\newtheorem{assumption}{Assumption}

\usepackage{verbatim}
\usepackage{subcaption}
\usepackage{tikz} 
\usetikzlibrary{arrows.meta}

\title{Policy Optimization in Multi-Agent Settings under Partially Observable Environments
}
\author{Ainur Zhaikhan, Malek Khammassi, and Ali H. Sayed

\thanks{A. Zhaikhan, M.Khammasi, and A. H. Sayed are with the Adaptive
 Systems Laboratory, \'Ecole Polytechnique F\'ed\'erale de Lausanne (EPFL),
 CH-1015, Switzerland. Emails: ainur.zhaikan@epfl.ch and ali.sayed@epfl.ch}}

\begin{document}
\maketitle
\begin{abstract}
This work leverages  adaptive social learning to estimate partially observable global states in  multi-agent reinforcement learning (MARL) problems. Unlike existing methods, the proposed approach enables the concurrent operation of social learning and reinforcement learning. Specifically, it alternates between a single step of social learning and a single step of MARL, eliminating the need for the time- and computation-intensive two-timescale learning frameworks. Theoretical guarantees are provided to support the effectiveness of the proposed method. Simulation results verify that the performance of the proposed methodology  can approach that of reinforcement learning when the true state is known.
    \end{abstract}
\section{Introduction}
The apparent benefits of collaboration are making multi-agent systems and multi-agent algorithms increasingly popular. However, multi-agent problems come with their own challenges, one of which is the limited observability of individual agents. Physical barriers, technical limitations of equipment, and noisy environments all contribute to the difficulty agents face in accurately determining the global state. Addressing partial observability in a multi-agent setting is known to be an NEXP-hard problem, as it requires each agent to have access to the observations of all other agents---an approach that is often impractical~\cite{NEXP}.\par 
Current solutions for partial observability in MARL are largely based on neural networks, where the responses of the environment must be inferred from local observations. This setup assumes that the local policies of agents are conditioned solely on their own observations. However, optimal policies should ideally be conditioned on the global state. Although such black-box approaches have shown success in certain applications, they often lack interpretability and theoretical grounding. To address this lack of clarity, the idea of using adaptive social learning for state estimation has been proposed in~\cite{Policy_eval}. \par 
In the framework of social learning, each agent estimates its own belief about the global state using local observations and available models for likelihood functions. The agents then diffuse their belief vectors with their neighbors to improve their estimates, employing an adaptation parameter that enables tracking of a dynamically changing state. Repeated application of these two operations leads to the convergence of all agents to the \textit{true belief vector}—that is, a basis vector with a unit entry at the location corresponding to the true state.\par 
In this paper, we aim to incorporate adaptive social learning to estimate states in MARL with partially observable environments. The key difference from ~\cite{Policy_eval} lies in the underlying reinforcement learning setting; while ~\cite{Policy_eval} focuses on a policy evaluation setting, we consider a more complex MARL framework, one that deals  with a multi-agent off-policy actor-critic (MAOPAC) formulation in a manner that extends the earlier work \cite{SUTTLE20201549}; this last work, however, assumes completely observable states while we forgo this requirement and study environments with \textit{only partially} observable global states. \par 
Specifically, in policy evaluation scenarios, the policy is fixed and the objective is to evaluate state values. In contrast, the MAOPAC scenario deals with policy {\em optimization}, where the goal is not only to evaluate but also to learn the optimal policy. This objective introduces an additional layer of complexity and randomness. Moreover, by being an off-policy algorithm, MAOPAC requires correction mechanisms such as importance sampling, which introduce additional intermediate variables. Consequently, analyzing the impact of state estimation on MAOPAC becomes significantly more complex.\par
The use of social learning for state estimation in MAOPAC has also been considered in~\cite{Ainur} with one key difference. In that work,  each iteration of the reinforcement learning algorithm involves multiple (repeated) rounds of social learning, which adds complexity. The use of repeated social interactions per iteration is meant to ensure that the social learning algorithm is able to attain a good estimate of the unobservable state. However, and interestingly, in this work we establish that social and reinforcement learning can work hand-in-hand and in parallel, with both converging together to the desired state. In other words, although the initial state estimates by social learning are not accurate due to limited interactions, the reinforcement learning strategy can still rely on these predictions and both methods will evolve over time in a manner that converges toward a stable solution.  We therefore end up with a parallel learning framework, where state estimation and reinforcement learning are performed simultaneously—specifically, each step of reinforcement learning is coupled with a single step of social learning.\par 
We show that the parallel implementation of social learning within MAOPAC can achieve $\varepsilon$-optimality under the condition that the underlying states change slowly. Therefore, our work is well-suited for   scenarios such as networks of maintenance machines or sensor networks operating in slowly changing environments.\par 
The main contributions of this work can be summarized as follows: 1) we extend the framework of MAOPAC to settings where the global state is only partially observable, 2) we investigate the use of adaptive social learning in MARL for more sophisticated reinforcement learning algorithms than those considered in prior works, 3) we integrate social learning with MAOPAC by removing the need for two-time-scale learning and instead implement a concurrent learning scheme, where global state estimation and reinforcement learning are carried out in parallel under the assumption of slowly evolving states, 4) we provide a theoretical analysis showing that the performance of the proposed scheme can approach that of MAOPAC under full observability, and 5) we illustrate the proposed method through simulation experiments.\par 
The organization of the paper is as follows. Section~\ref{Sec:state of the art} presents the state of the art. Section~\ref{sec: Preliminary} introduces the relevant notation and definitions. Section~\ref{sec:algorithm} describes the proposed method. Theoretical analysis and simulation results are presented in Sections~\ref{sec:theory} and ~\ref{sec:experiment}, respectively. Concluding remarks are provided in Section \ref{sec:conclusion}.

\section{State of the art}
\label{Sec:state of the art}
\subsection{Pseudo-decentralized solutions}
Due to the challenges of achieving full decentralization under partial observability, several algorithms have been proposed that adopt pseudo-decentralized approaches. One such approach is \textit{independent learning}~\cite{IL1}, where each agent learns independently by treating the behavior of other agents as part of a non-stationary environment. While this method can be effective in some scenarios, it often struggles with environmental instability in more complex tasks.\par 

Another prominent class of algorithms is based on the \textit{centralized training with decentralized execution} (CTDE) paradigm. This includes methods such as \textsc{QMIX}~\cite{QMIX}, \textsc{COMA}~\cite{COMA}, \textsc{MADDPG}~\cite{MADPPG}, \textsc{VDN}~\cite{VDN}, and \textsc{QTRAN}~\cite{QTRAN}. As the name suggests, CTDE methods leverage a centralized component during training—typically to estimate global state values—while enabling agents to act using decentralized policies during execution. In our work, we aim to overcome the limitations of CTDE by enabling \textit{fully decentralized learning} through \textit{inter-agent communication}, eliminating the need for a central unit during training.
\subsection{Neural network-based approaches}  
A common strategy for addressing multi-agent reinforcement learning (MARL) under partial observability involves leveraging neural network (NN) architectures~\cite{Neural_POMDP, Actor_critic_POMDP, POMDP_RNN, Mean_field}. Several NN-based methods, such as those presented in~\cite{Actor_critic_POMDP, Neural_POMDP, Belief1, LSTM, NN_POMDP}, have been successfully applied to single-agent problems with partial observability, while extensions to multi-agent settings are explored in~\cite{POMDP_RNN, NN_model_free, NN_model_free2}. These methods typically rely on neural networks to learn direct mappings from observations to observation-action values.

However, substituting state-action values with observation-action values does not always yield optimal performance, as the success of such methods depends on the model’s ability to accurately approximate the underlying value functions. Moreover, in NN-based approaches, local policies are usually conditioned directly on local observations, whereas, in theory, individual policies should ideally be conditioned on the global state.

The scheme proposed in our work allows us to  condition policies on a \textit{belief vector}, which represents a probabilistic estimate of the global state. Additionally, NN-based techniques are often regarded as brute-force solutions, offering limited analytical insight. Their evaluation is largely based on empirical performance metrics—such as computational efficiency and network complexity—rather than theoretical guarantees. In contrast, our work addresses partial observability through adaptive social learning, which provides theoretical convergence guarantees. This allows our proposed approach to be effective in practice, as well as supported by theoretical foundations.


\subsection{Communication-based solutions}  
Prior research has explored the use of belief diffusion to tackle decentralized partially observable Markov decision processes (Dec-POMDPs). Studies such as \cite{Consensus1, Similar_1} investigate how consensus protocols can be adapted for Dec-POMDPs, primarily in the context of online planning. However, their application to reinforcement learning remains largely unexplored. These works mainly analyze consensus algorithms without establishing specific convergence proofs within RL settings and often rely on heuristic-based methodologies.  

The study in \cite{Consensus2} integrates belief diffusion into fundamental RL methods like Q-learning. Our approach embeds belief diffusion into more sophisticated RL techniques—specifically, off-policy multi-agent actor-critic frameworks—while providing rigorous analysis and convergence conditions.
The use of adaptive social learning in the context of reinforcement learning was first considered in \cite{Policy_eval}, where the authors focused on the evaluation of a fixed (not necessarily optimal) policy. In contrast, our work addresses the more challenging problem of learning an {\em optimal} policy. Moreover, our proposed method involves off-policy learning, which requires the incorporation of correction terms. This leads to the introduction of additional variables and necessitates a more complex theoretical analysis compared to policy evaluation.
Another work that  leverages communication among agents and proposed a fully decentralized scheme for MARL  with partial observations  is presented in \cite{MA_POMDP}. This method employs zeroth-order policy optimization (ZOPO) to localize learning by decomposing cumulative reward estimation steps and then diffusing those estimates with neighbors. This method has theoretical guarantees. However,  the method is designed for on-policy learning which is less general than MAOPAC. Moreover, the scheme is computationally costly since each update of  policy parameters requires knowledge of cumulative rewards, which requires another inner loop of experiments for estimation as in the REINFORCE algorithm \cite{REINFORCE}. 

\section{Multi-Agent Model}
\label{sec: Preliminary}
We model the setting as a \textit{Decentralized Partially Observable Markov Decision Process (Dec-POMDP)} defined by the tuple:
\begin{align}
\mathcal{M} \triangleq \left(\mathcal{K}, \widetilde{\mathcal{A}}, \mathcal{S}, \{\mathcal{O}_{k}\}_{k=1}^{K}, \{r_{k}\}_{k=1}^{K}, \mathcal{P}, \{\mathcal{L}_{k}\}_{k=1}^{K}\right)
\end{align}
\noindent where the symbols denote the following variables:
\begin{itemize}
    \item \( \mathcal{K} \triangleq \{1, 2, \ldots, K\} \): the set of agents.
    \item \( \mathcal{A} \): the action set for each agent.
    \item \( \widetilde{\mathcal{A}} \triangleq \mathcal{A}^K \): the joint action space.
    \item \( \mathcal{S} \): the set of global states.
    \item \( \mathcal{O}_k \): the observation set for agent \( k \in \mathcal{K} \).
    \item \( \mathcal{L}_k(\xi | s) \): a likelihood function specifying the probability of agent \( k \) observing \( \xi \in \mathcal{O}_k \) given the true state \( s \in \mathcal{S} \).
\end{itemize}
\noindent We assume that the action space \( \mathcal{A} \) and the state space \( \mathcal{S} \) are finite, with sizes \( A \) and \( S \), respectively.\par 
The environment dynamics is governed by a transition model
$\mathcal{P} : \mathcal{S} \times \widetilde{\mathcal{A}} \to \mathcal{S}$. 
Specifically, $\mathcal{P}(s_2 | s_1, a)$ denotes the probability of transitioning from state \( s_1 \) to state \( s_2 \) under a joint action \( a = \{a_k\}_{k=1}^K \). Each agent \( k \) has an individual reward function \( r_k : \mathcal{S} \times \mathcal{A} \times \mathcal{S} \to \mathbb{R} \), assumed to be uniformly bounded by a constant \( R_{\max} > 0 \). The reward received by agent \( k \) at time \( n \), due to a transition from state \( s_n \) to \( s_n' \) under action \( a_{k,n} \), is denoted by:
\begin{align}
r_{k,n} \triangleq r_k(s_n, a_{k,n}, s_n')
\end{align}
We assume a networked communication framework among agents. The communication is defined by a graph $G$, where each agent corresponds to a node. Two agents can exchange information if a direct (single-hop) edge exists between them in the graph. \par 
We consider a $K\times K$ combination matrix 
$C$, which assigns weights to neighboring agents during information exchange. The matrix $C$ is \textit {doubly stochastic}, meaning that each row and column sums to one. Moreover, it respects the structure of the graph 
$G$; specifically, if there is no edge between agents 
$i$ and $j$, then the corresponding entry $c_{i,j}$ is zero. On the other hand, if agents $i$ and $j$ are linked, then the weights on the edges connecting them in both directions are $c_{ij}$ and $c_{ji}$. Furthermore, we assume that the underlying network topology is strongly connected, as formally stated in Assumption~\ref{assumption:graph}.
 \par 
\begin{assumption}[\textbf{Strongly connected graph}]
\label{assumption:graph}
    The underlying graph topology is assumed to be \textit{strongly connected}, i.e, any two agents in the network can be connected via a path with positive combination weights and, at least one agent $\ell$ in the network has a nonzero self loop, $c_{\ell,  \ell}>0$. \qed
\end{assumption}
With respect to the multi-agent policy model, we assume that each agent’s decisions are conditionally independent of the actions of other agents, given the current global state. Under this assumption, any joint policy \( \pi \)  can be factorized as:
\begin{align}
    \pi(\tilde{a} \mid s) &= \prod_{k=1}^K \pi_{\ell}(a_{k} \in \mathcal{A} | s), \quad \forall \widetilde{a} \in \widetilde{\mathcal{A}}, \forall s\in \mathcal{S} \label{eq:policy_model}
\end{align}
where $\pi_{k}$ denotes the individual policy executed by agent $k$ and $\tilde{a}\triangleq \{a_1, a_2, \dots a_K \} \in \widetilde{\mathcal{A}}$ denotes the joint action taken by all agents in the network. This model is adopted to facilitate fully decentralized learning.



\section{MARL+Social learning}
\label{sec:algorithm}
This work aims to study the effectiveness of multi-agent reinforcement learning (MARL) in settings where the global state is not fully observable and is inferred by means of a social learning strategy. Henceforth, we refer to the proposed scheme as \textit{MARL+SL}.
\par 
We focus on a specific MARL setting, namely, Multi-Agent Off-Policy Actor-Critic (MAOPAC) introduced in \cite{SUTTLE20201549}. However, MAOPAC assumes that the global state is known by all agents. While actor-critic approaches are powerful tools with strong theoretical guarantees, their reliance on global state information greatly limits their practical applicability. In real-world scenarios, agents often operate in large, spatially distributed environments where full observability is unrealistic. Instead, each agent typically has access only to local observations, which provide partial and potentially limited information about the global state.

To address this limitation, the scheme proposed in our work enables agents to estimate the global state by leveraging both their local observations and inter-agent communication. This allows each agent to operate effectively in settings with limited observability, broadening the practical utility of actor-critic methods in decentralized environments.

For this purpose, we will devise a framework that integrates MAOPAC with adaptive social learning to effectively predict the unknown global state in partially observable environments. One of the main features of our approach is that it  enables state estimation and reinforcement learning to operate on a single time-scale. That is, state learning and policy learning occur concurrently, eliminating the need for multiple iterations of state estimation for each reinforcement learning update.
To achieve this objective, we first provide a brief overview of the MAOPAC algorithm and then describe how social learning can be leveraged to enable MAOPAC in partially observable settings.  The full listing of the proposed scheme is provided later in Algorithm \ref{Algorithm1}, which will be referenced repeatedly in subsequent discussions.

\subsection{MAOPAC} 
The MAOPAC algorithm studied in ~\cite{SUTTLE20201549} belongs to the class of actor-critic methods~\cite{RL_book}, which are among the most widely adopted techniques in reinforcement learning due to their well-established theoretical properties. At a high level, MAOPAC simultaneously learns individual optimal policies \( \pi^{\dagger}_{k} \) using policy gradient techniques, while estimating and utilizing global state-value functions to guide learning. The MAOPAC formulation from \cite{SUTTLE20201549} is an advanced version of the multi-agent actor-critic algorithm that addresses possible instabilities arising from function approximations. To achieve this, it introduces additional intermediate variables such as \textit{emphatic weightings} and \textit{eligibility traces} \cite{ETD1, ETD2, SUTTLE20201549}. \par 
For simplicity of analysis, and without significant loss of generality, we will consider a special case of MAOPAC by removing the emphatic weighting and eligibility trace variables, thereby reducing the algorithm to the classical actor-critic framework. This simplification remains effective when linear function approximation is suitable for the problem and the risk of instability arising from approximation errors is minimal.\par 
To formalize the method, we associate  a state-value function with a joint policy \( \pi \), and define it for each state \( s \in \mathcal{S} \) as
\begin{align}
    V^\pi(s) \triangleq \mathbb{E}\left[\sum_{n=0}^{\infty} \gamma^n \bar{r}_n \mid s_0 = s \right].
\end{align}
where  \( \gamma \in (0,1) \) is the discount factor and $\bar{r}_n$ is the average reward across all \( K \) agents at time step \( n \), i.e., 
\begin{align}
    \bar{r}_n = \frac{1}{K} \sum_{k=1}^{K} r_{k,n},
\end{align}
MAOPAC aims  to learn a joint policy \( \pi^{\dagger} \triangleq \prod_{k=1}^K\pi_{k}^{\dagger}\) that maximizes the expected cumulative reward over an infinite horizon, i.e., 
\begin{align}
   \pi^{\dagger} =\arg \max \limits_{\pi} J
\end{align}
where 
\begin{align}
    J \triangleq \mathbb{E}\left[\sum_{n=0}^{\infty} \gamma^n \bar{r}_n \right]=\mathbb{E}_{s} V^{\pi}(s),
\end{align}
The optimal policy estimate by an arbitrary agent $k$ at time $n$ will be denoted by $\pi_{k,n}$ and will be approximated by some function $g$  of the state feature vector $\mu_n$ corresponding to the current state $s_n$ and parameter vectors $\theta_{k,n} \triangleq \{\theta_{k,n}^a \in \mathbb{R}^S\}_{a \in \mathcal{A}}$ as follows 
\begin{align}
    \pi_{k,n}(a|s_n) \approx g (a|\mu_{n};\theta_{k,n}), \forall a \in \mathcal{A} \label{eq:target_pol_def}
\end{align}
An example of a policy function is the Boltzmann dsitribution, defined as:
\begin{align}
g(a|\mu;\theta) = \frac{\exp(\mu^{T}\theta^{a})}{\sum\limits_{a'\in \mathcal{A}}\exp(\mu^{T}\theta^{a'})} \label{eq:boltzman}
\end{align}
where $\theta\triangleq \{\theta^a\}_{a\in \mathcal{A}}$. \par 
Since we assume an off-policy learning setting, the behavioral policy of each agent, denoted by $b_{k,n}$, may differ from the estimated optimal policy $\pi_{k,n}$. We approximate the behavioral policies using a function $q$, parameterized by vectors $\chi_k \triangleq \{\chi_{k}^{a} \in \mathbb{R}^{S}\}_{a \in \mathcal{A}}$:
\begin{align}
     b_k(a|s_n) \approx  q(a|{\mu}_{n};\chi_{k}), \forall a\in \mathcal{A} \label{eq:beh_pol_def}
\end{align}
Off-policy learning requires statistical corrections using the \textit{joint importance sampling ratio} defined as 
\begin{align}
\rho_{n}\triangleq \prod \limits_{k\in \mathcal{K}}\frac{ \pi_{k,n}(a_{k,n}|s_n)}{b_{k}(a_{k,n}|s_{n})}\approx \prod \limits_{k\in \mathcal{K}} \frac{g (a_{k,n}|\mu_{n}; \theta_{k,n})}{q_{k}(a_{k,n}|\mu_n; \chi_k )} \label{eq:joint_rho}
\end{align}
To learn the optimal policies, each agent maintains a local estimate of the state value function, which is approximated using the following linear parameterization:
\begin{align}
\label{eq:linear_approx}
v_{\omega_{k,n}} \approx \mu_{n}^T \omega_{k,n}
\end{align}
where \( \mu_{n} \in \mathbb{R}^S \) denotes the state feature vector corresponding to the state $s_n$, and \( \omega_{k,n} \in \mathbb{R}^S \) represents a parameter vector, commonly referred to as the \textit{critic parameter}.   Note that in the proposed scheme, state feature vectors $\mu_{n}$ will be redefined as \textit{belief vectors}; the reason for this reinterpretation will be explained in subsequent sections. \par 
Each agent updates its critic parameter using a local temporal difference error defined as
\begin{align} 
            &\delta_{k,n} = r_{k,n} + \gamma \omega_{k,n}^T \eta_{n} - \omega_{k,n}^T \mu_{n} \label{eq:concept_omega}
\end{align}
where $\eta_n$ denotes the state feature vector corresponding to the next global state $s_{n+1}$.  Therefore, the first step in estimating the critic parameters is to compute
\begin{align}
    \widetilde{\omega}_{k,n} = \omega_{k,n} + \beta_n \rho_n \delta_{k,n} \mu_{n} \label{eq:concept_delta1}
\end{align}
where the importance sampling ratio $\rho_n$ is applied to account for the distribution mismatch introduced by off-policy data collection.

Note that the updates in \eqref{eq:concept_omega}--\eqref{eq:concept_delta1} are based on local rewards. Therefore, to aggregate local information and improve estimation accuracy, agents share their individual critic values with direct neighbors and employ a diffusion strategy to obtain a more accurate estimate of the state value: 
       \begin{equation}
 \omega_{k,n+1} = \sum_{\ell \in \mathcal{N}_k} c_{\ell,k} \widetilde{\omega}_{k,n+1} \label{eq:concept_omega_dif}
        \end{equation}
The scheme in \eqref{eq:concept_delta1}--\eqref{eq:concept_omega_dif}, inspired by \cite{Sayed_diffusion} and first applied to state value estimation in \cite{SUTTLE20201549}, enables agents to estimate the global state value by means of inter-agent communication, despite having access only to local rewards.
\noindent 

\par 
 The critic parameters are estimated in the MAOPAC method, as they are involved in the learning of individual optimal policies. As demonstrated in \cite{SUTTLE20201549}, the gradient of the objective function $J$ with respect to the policy parameters $\theta_{k,n}^a$ can be approximated using samples of
 \begin{align}
     \mathbb{E}\left[\rho_n \delta_{k,n} \nabla_{\theta_{k,n}^a} \log \pi_{k,n}\left(a_{k,n} \mid s_n\right)\right]
 \end{align}
 Therefore, for all $a\in \mathcal{A}$, the policy parameters $\theta_{k,n}^a$, also referred as \textit{actor parameters}, are updated as:
\begin{align}
    \theta_{k,n+1}^a = \theta_{k,n}^a + \beta_{\theta,n} \rho_n \delta_{k,n} \nabla_{\theta^{a}_{k,n}} \log g(a_{k,n} \vert \mu_{n}) \label{eq:}
\end{align}        
\subsubsection{Estimation of the importance sampling ratio}
Note that the variable $\rho_n$ requires knowledge of all behavioral and target policies across the entire network. However, due to the decentralized setting, each agent has access only to its \textit{individual sampling ratio} $\rho_{k,n}$, defined as:
\begin{align}
    \rho_{k,n}&\triangleq\frac{ \pi_{k}(a_{k,n}|s_n)}{b_{k}(a_{k,n}|s_{n})}\approx \frac{g(a_{k,n}|{\mu}_{n}; \theta_{k,n})}{q(a_{k,n} {\mu}_n; \chi_{k})}  \nonumber \\
    &\triangleq h(a_{k,n}|\mu_{n};\theta_{k,n},\chi_k) \label{eq:ind_rho}
\end{align}
We assume that the individual importance sampling ratios are bounded, and we denote their minimum and maximum values by $\rho_{\min}$ and $\rho_{\max}$, respectively:  
\begin{align}
    \rho_{\min}\leq \rho_{k,n} \leq {\rho_{\max}}, \quad \forall k \in \mathcal{K}
\end{align}
To maintain a fully decentralized learning process, the work ~\cite{SUTTLE20201549} uses consensus updates to estimate the joint importance sampling ratio. Specifically, agents exchange information with their direct neighbors and apply the following update to the logarithm of their individual importance sampling ratios, repeated for $t = 0, 1, \ldots, T$:
\begin{align}
    f_{k,n}^{t+1}=\sum_{\ell \in \mathcal{N}_{k}} c_{k,\ell} f_{\ell,n}^t \label{eq:rho_consensus}
\end{align}
where 
\begin{align}
    f_{k,n}^0\triangleq \log \rho_{k,n} \nonumber
\end{align}
After sufficiently large $T$, the estimates at all agents converge to the same network-wide value; that is, for any \( k \neq j \in \mathcal{K} \), it holds that \cite{Average1, Average2, SUTTLE20201549}:
    \begin{equation}
f_{k,n}^T=f_{j,n}^T= \frac{1}{K} \sum_{k =1}^K f_{k,n}^0
\end{equation}
Then, each agent can retrieve the joint importance sampling ratio as follows:
\begin{equation}
\exp \left(K f_{k,n}^T\right)=\exp \left(\sum_{k=1}^K f_{k, n}^0\right)=\prod \limits_{k\in \mathcal{K}} \rho_{k, n}= \rho_n 
\end{equation}
From the description so far, it is evident that MAOPAC requires access to the global state, specifically the global state feature vector $\mu_n$. Consequently, MAOPAC in its original form is not applicable to environments with partial observability. In the following, we employ social learning that addresses this limitation. By leveraging social learning, we enable MAOPAC to function effectively in partially observable environments in a fully decentralized manner.
\subsection{State estimation using social learning}
Social learning is a framework in which networked agents collaborate to identify the hypothesis that best explains their observations ~\cite{JADBABAIE2012210, Nedic2017, ASL, SL_book}. For example, consider a scenario where the global state represents the position of some object on a grid. Agents do not see the object's exact location—i.e., the true global state is unknown—but they receive observations \(\xi_{k,n}\) related to the current global state \(s_n\). For instance, instead of knowing the exact location of the object, agents may be aware of their own coordinates and receive a noisy distance measurement to the object, which serves as their observation. Individually, the available information may be insufficient for a single agent to accurately infer the exact location of the object. However, through collaborative information sharing, the entire network of agents can achieve a more precise estimation of the true global state.

Technically, each agent maintains a belief vector \(\widetilde{\mu}_{k,n} \in \mathbb{R}^{S}\), where the \(i\)-th element represents the probability that the true state is \(i\). Upon receiving observations \(\xi_{k,n}\), agent $k$ updates its belief vector using the likelihood function \(\mathcal{L}_k(\xi_{k,n} | s)\) as described below:
\begin{align}
&\psi_{k, n}(s)  \sim \mathcal{L}_k\left({\xi}_{k,n} \vert s\right) {\widetilde{\mu}}_{k, n}(s)  \label{belief_update_trad}
\end{align}
\par 
\noindent Equation \eqref{belief_update_trad} is simply a local Bayesian update, where the prior probability \(\widetilde{\mu}_{k, n}\) is multiplied by the likelihood function to obtain the posterior probability vector $\pi_{k,n}(s)$; the symbol $\sim$ is used to indicate that the entries of the posterior vector are normalized to add up to one. The posterior estimates \(\psi_{k,n}\) are exchanged with direct neighbors, and belief vectors are refined through geometric averaging, integrating all available information from neighboring agents, as shown below:
\begin{align}
    \widetilde{\mu}_{k, n}(s) \sim \prod_{\ell \in \mathcal{N}_k}\left[{\psi}_{\ell, n}(s)\right]^{c_{\ell k}} \label{belief_update_trad2}  
\end{align}
where $\mathcal{N}_k$ denotes the set of agents that are neighbors of agent $k$.  In the subsequent iterations, new observations arrive, and the update steps \eqref{belief_update_trad}-\eqref{belief_update_trad2} are repeatedly applied to refine the current estimate until all belief vectors reach consensus. \par 
The update scheme in \eqref{belief_update_trad}--\eqref{belief_update_trad2} is well-suited for scenarios where the global state remains fixed. However, in reinforcement learning, the state evolves over time. To address this, one approach is to repeatedly apply steps \eqref{belief_update_trad}--\eqref{belief_update_trad2} at each reinforcement learning iteration until a satisfactory state estimate is obtained, as explored in~\cite{Ainur}. While this method can be effective in environments with rapidly changing states, it introduces a two-time-scale learning dynamic, which increases computational complexity.\par 
In this work, we employ adaptive social learning (ASL)~\cite{ASL,SL_book}, which provides efficient state estimation in dynamic environments, as long as the global state evolves gradually.
 Some useful studies on the ability of adaptive social learning to track slow Markov chains are \cite{Malek,Malek2}. In ASL, the elements of the belief vector \(\psi_{k,n}(s)\) are updated using the following scheme:

\begin{align}
&\psi_{k, n}(s)  \sim \mathcal{L}_k\left({\xi}_{k,n} \vert s\right) {\widetilde{\mu}}^{1-\sigma}_{k, n}(s)\label{belief_update} \\
&\widetilde{\mu}_{k, n}(s) \sim 
 \prod_{\ell \in \mathcal{N}_k}\left[{\psi}_{\ell, t}(s)\right]^{c_{\ell k}}\label{belief_update2}
\end{align}  
This scheme is similar to \eqref{belief_update_trad}-\eqref{belief_update_trad2}, except for the introduction of the positive adaptation parameter \(\sigma\). The parameter $\sigma$ regulates the trade-off between prior and newly acquired information in the belief update process. Specifically, it controls the contribution of the prior belief $\widetilde{\mu}^{1-\sigma}_{k,n}$ relative to the likelihood of the new observation. A larger value of $\sigma$ downgrades the influence of the prior, making the algorithm more responsive to recent observations and thus more adaptable to changes in the underlying state. Conversely, a smaller $\sigma$ places more emphasis on the prior, allowing the algorithm to exploit accumulated information and maintain accuracy in stationary environments. Consequently, $\sigma$ plays a critical role in managing the state estimation error when the state evolves over time.\par 
In traditional MAOPAC, the state feature vectors \( \mu_{n} \) can take any values within the range \([0,1]\). Therefore, in our approach, we replace \( \mu_{n} \) in \eqref{eq:target_pol_def}--\eqref{eq:ind_rho} with a belief vector \( \mu_{k,n} \), which is obtained by first applying the update steps \eqref{belief_update} and \eqref{belief_update2}, followed by a \textit{hard assignment}, defined by:
\begin{align}
    \mu_{k,n}(s) =
\begin{cases} 
1, & s =\arg\max\limits_{s'\in \mathcal{S}} \widetilde{\mu}_{k,n}(s^{\prime}), \\
0, & \text{otherwise}.  
\end{cases} \label{eq:int13}
\end{align}
The reason for performing the additional operation on the belief vectors, rather than directly using $\widetilde{\mu}_{k,n}$, will be discussed in the following subsection.\par 
Each update in MAOPAC requires knowledge of both the current and next states. We estimate the next state \( s_n^{\prime} \) in the same manner as the current state $s_n$. To achieve this, we introduce the belief vector \( \eta_{k,n} \), which is updated using the same scheme as in \eqref{belief_update}-\eqref{eq:int13}, i.e., for all \( s \in \mathcal{S} \) and \( k \in \mathcal{K} \)
\begin{align}
\zeta_{k, n}(s) & \sim \mathcal{L}_k\left(\xi_{k, n+1} \vert s\right) \widetilde{\eta}_{k, n}^{1-\sigma}(s) \label{eq:eta1}\\
\widetilde{\eta}_{k, n}(s) & \sim \prod_{\ell \in \mathcal{N}_k}\left[\zeta_{\ell, n}(s)\right]^{c_{\ell k}}
\end{align}
 \begin{align}
 \label{eq:eta3}
    \eta_{k,n}(s) =
\begin{cases} 
1, & s =\arg\max\limits_{s'\in \mathcal{S}} \widetilde{\eta}_{k,n}(s'), \\
0, & \text{otherwise}.  
\end{cases} 
\end{align}
where $\sigma$ is chosen as in \eqref{eq:int12}. 
At first glance, it may seem that state estimation is performed twice in each iteration. However, this is only true for the first iteration. In subsequent iterations, the current state is simply assigned as the next state from the previous iteration (step \eqref{alg:assign} in Algorithm \ref{Algorithm1}). Thus, state estimation actually occurs only once per iteration.

The full listing of the MARL+SL method—including the replacement of the feature vectors $\mu_n$ with belief vectors $\mu_{k,n}$, along with the steps for state estimation—is given in Algorithm~\ref{Algorithm1}. 
\begin{algorithm}
\caption{MARL+SL strategy for POMDPs}
\label{Algorithm1}
\begin{algorithmic}[1]
\small
\State \textbf{Initialize parameters:} $\lambda\in (0,1)$, $\zeta\in (0,1)$, $\gamma\in (0,1)$, $\omega_{k,0}(s)=0$, $\theta_{k,0}^a(s)=0$, $\rho_{k,0}$, $\widetilde{\mu}_{k,0} = \frac{1}{S}$, $\eta_{k,0} = \frac{1}{S}$, $\forall s \in \mathcal{S}, k \in \mathcal{K}$
\For{$n = 0, 1, 2, \dots$}
    \State Each agent $k$ takes action $a_{k,n} \sim q(a |\mu_{k,n})$
    \State Each agent $k$ receives reward $r_{k,n}$
    \State Each agent $k$ receives observation $\xi_{k,n+1}$
    \State Each agent $k$ estimates $\eta_{k,n}$ using equations \eqref{eq:eta1}--\eqref{eq:eta3}
    
    \ForAll{agents $k$}
        \begin{align}
        &\delta_{k,n} = r_{k,n} + \gamma \omega_{k,n}^T \eta_{k,n} - \omega_{k,n}^T \mu_{k,n}
        \label{alg:delta}\\
        &\widetilde{\omega}_{k,n} = \omega_{k,n} + \beta_n \rho_n \delta_{k,n} \mu_{k,n}
        \label{alg:omega1}\\
        &\Psi_{k,n}^{a} = \nabla_{\theta^{a}_{k,n}} \log g(a_{k,n} \vert \mu_{k,n}), \forall a \in \mathcal{A}
        \label{alg:Psi}\\
        &\theta_{k,n+1}^a = \theta_{k,n}^a + \beta_{\theta,n} \rho_n \delta_{k,n} \Psi_{k,n}^{a}, \forall a \in \mathcal{A}
        \label{alg:theta}
        \end{align}
    \EndFor

    \ForAll{agents $k$}
        \begin{equation}
        \hspace*{-1.71cm}\omega_{k,n+1} = \sum_{\ell \in \mathcal{N}_k} c_{\ell,k} \widetilde{\omega}_{\ell,n}
        \label{alg:omega2}
        \end{equation}
    \EndFor

    \ForAll{agents $k$}
        \begin{equation}
        \hspace*{-3.4cm} \mu_{k,n} = \eta_{k,n}
        \label{alg:assign}
        \end{equation}
    \EndFor
\EndFor
\end{algorithmic}
\end{algorithm} 
\subsection{Error probability measure}
Any estimation inevitably involves some errors. In our case, we estimate the global state, which in turn influences the updates of other reinforcement learning variables. Therefore, it is crucial to ensure that the state estimation error remains sufficiently small to guarantee that the resulting policies remain near-optimal with a certain level of accuracy. To this end, in the following, we analyze the error probability associated with the state estimation procedure described in \eqref{belief_update}--\eqref{eq:int13}.\par 
Due to the assumption that the agents' behavioral policies \( b_k \) are time-invariant and the environment follows a Markov decision process, the state transitions form a Markov chain \cite{Sutton}. Therefore, for any two distinct global states \( s, s' \in \mathcal{S} \), the state transition dynamics are modeled as:
\begin{equation}
\mathbb{P}\left[s_n=s^{'} \vert s_{n-1}=s\right]=\left\{\begin{array}{ll}
1-\varepsilon q_{s s}, & \text { if } s=s^{'} \\
\varepsilon q_{s s^{'}}, & \text { if } s \neq s^{'}
\end{array}\right.
\end{equation}
where \(\varepsilon\) denotes a \textit{drift parameter}, which defines the pace of state changes, i.e., $0<q_{ss}<1/{\varepsilon}$ and 
\label{eq:trans_prob_MC}
\begin{align}
    \sum_{s\neq s^{'}} q_{ss'}= q_{ss}
\end{align}
We also introduce the notation for the error probability 
\(p_{k,n}\) defined by 
\begin{align}
p_{k, n} \triangleq \mathbb{P}\left[\arg \max _{s \in \mathcal{S}} \widetilde{\mu}_{k, n}(s) \neq s_n\right]
\end{align}
The work \cite{Malek2} established that, under Assumptions~\ref{assumption:belief_init}–\ref{assumption:statistical_model} stated below, and provided the parameter $\sigma$ satisfies
\begin{align}
\sigma = \frac{\nu}{\log{\frac{1}{\epsilon}}}, \label{eq:1int12}
\end{align}
for some constant $\nu$ such that $0 < \nu < \Phi_{\mathrm{min}}$, where $\Phi_{\mathrm{min}}$ denotes the error exponent defined in~\cite{Malek2}, the asymptotic error probability satisfies the bound
\begin{align}
\limsup_{n \rightarrow \infty} p_{k,n} \leq \kappa~ \varepsilon \log{\frac{1}{\varepsilon}} + o\left(\varepsilon \log{\frac{1}{\varepsilon}}\right), \label{eq:1int11}
\end{align}
where $\kappa$ is a constant also characterized in~\cite{Malek2}.

Equations~\eqref{eq:1int12}--\eqref{eq:1int11} show that,
with an appropriate choice of the adaptation parameter~$\sigma$, the state‐estimation error can be kept tightly bounded when the drift parameter~$\varepsilon$ is small. This property is particularly useful, as discussed further, for managing the estimation error in critic parameters $\omega_{k,n}$.   
\par
The results in \eqref{eq:1int12}--\eqref{eq:1int11} represent the asymptotic bound for \( p_{k,n} \), while the exact expression for the non-asymptotic bound can be found in \cite{Malek2}. For clarity of presentation, we do not reproduce the lengthy expressions here and instead denote the bound by \( B_{\mu,n}(\sigma, \varepsilon) \), which depends on the time index \( n \), the adaptation parameter \( \sigma \), and the state dynamics parameter \( \varepsilon \). To simplify notation, we suppress the dependence on \( \sigma \) and \( \varepsilon \), and refer to the bound as \( B_{\mu,n} \).

\begin{assumption}[\textbf{Belief vector initialization}]
\label{assumption:belief_init}
For all agents \( k \in \mathcal{K} \), all states are equiprobable initially, i.e.,  
\begin{align}
    \widetilde{\mu}_{k,0}(s) = \frac{1}{S}, \quad \forall s \in \mathcal{S}.
\end{align}
\end{assumption}

\begin{assumption}[\textbf{Bounded likelihood functions}]
There exists a constant $C$, such that for all $s \in \mathcal{S}$ and for all $k \in \mathcal{K}$
\begin{align}
&\sup \limits_{\xi \in \mathcal{O}_{\ell}}\left|\log \frac{\mathcal{L}_k\left(\xi \vert s^{\prime}\right)}{\mathcal{L}_k(\xi \vert s)}\right| < C
\end{align} 
\end{assumption}
\begin{assumption}[\textbf{Global identifiability}]
For any two distinct states $s\in \mathcal{S}$ and $s^{\prime}\in \mathcal{S}$ there is at least one agent $k$ for which 
    \begin{align}
&\mathbb{E}_{s^{\prime}}\left[\log \frac{\mathcal{L}_k\left(\xi \vert s^{\prime}\right)}{\mathcal{L}_k(\xi \vert s)}\right]>0
    \end{align}
\end{assumption}
\begin{assumption}[\textbf{Statistical model}]
\label{assumption:statistical_model}
Let \( \xi_n \triangleq \{\xi_{k,n}\}_{k=1}^{K} \) and let \( \mathcal{L} \) denote the joint likelihood function. Conditioned on the value of the state at time $n$, the current observations $\xi_n$ are independent of any past states or observations and independent over space. Therefore, the joint likelihood at time $n$ satisfies
\begin{align}
&\mathcal{L}\left(\xi_n \vert s_n, \ldots, s_0, \xi_{n-1}, \ldots, \xi_0\right)=\mathcal{L}\left(\xi_n \vert s_n\right)\\
&\mathcal{L}\left(\xi_n \vert s_n\right)=\prod_{k=1}^K \mathcal{L}_k\left(\xi_{k, n} \vert s_n\right)
\end{align} \qed
\end{assumption}
Let \(\mu_n^{\star}\) denote the true belief vector at time \(n\), where all elements are zero except for the one corresponding to the current true state \(s_n\), which is set to \(1\). We define the difference between the true belief vector and the estimated one as:  
\begin{align}
    \Delta \mu_{k,n} = \mu_n^{\star} - \mu_{k,n}
\end{align}
Since the critic and actor parameters are continuous-valued vectors, the analysis requires bounding the first- and second-order moments of the error,  
\( \mathbb{E}\|\Delta \mu_{k,n}\| \) and \( \mathbb{E}\|\Delta \mu_{k,n}\|^2 \), rather than relying solely on error probabilities \( p_{k,n} \). The hard assignment step \eqref{eq:int13} is introduced to facilitate the derivation of these bounds, which would otherwise be challenging to establish.
To illustrate this point, we define the following complementary events.
\begin{align}
    \phi_{k,n} \triangleq \left(\arg\max\limits_{s\in \mathcal{S}} \widetilde{\mu}_{k,n} = s_{n} \right),\label{eq:int14}\\
    \widetilde{\phi}_{k,n} \triangleq \left(\arg\max\limits_{s\in \mathcal{S}} \widetilde{\mu}_{k,n} \neq s_{n} \right). \label{eq:int15}
\end{align}
Since the events $\phi_{k,n}$ and $\widetilde{\phi}_{k,n}$ are complementary, we can apply the law of total expectation to obtain the following bound:   
\begin{align}
\mathbb{E} (\|\Delta \mu_{k,n}\|) &= \mathbb{E} (\|\Delta \mu_{k,n}\||\widetilde{\phi}_{k,n}) p_{k,n} \nonumber\\
& \quad + \mathbb{E} (\|\Delta \mu_{k,n}\||{\phi}_{k,n}) (1 - p_{k,n}) 
 \nonumber\\
&\stackrel{(a)}{=} \mathbb{E} (\|\Delta \mu_{k,n}\||\phi_{k,n}) p_{k,n} \leq \sqrt{2} p_{k,n}
\end{align}
where $(a)$ follows since $\|\Delta \mu_{k,n}\|=0$ given $\phi$ is true. Step $(a)$ would not hold if we directly use \( \widetilde{\mu}_{k,n} \) instead of \( \mu_{k,n} \), since obtaining a bound in terms of \( p_{k,n} \) would become nontrivial. \par 
Up to this point, we presented the proposed MARL+SL method and the error probability associated with adaptive social learning. In the next section, we introduce the main theoretical results that describe how state estimation using SL impacts other variables in the multi-agent reinforcement learning framework, and we verify the effectiveness of the proposed method.
\section{Theoretical results}
\label{sec:theory}
Before delving into the main theoretical results, we first outline several additional assumptions that underpin our analysis.
\begin{assumption}[\textbf{Learning rate}]
\label{assumption:learning rate}
The learning rates $\beta_n$ and $\beta_{\theta,n}$ used in \eqref{alg:omega1} and \eqref{alg:theta} satisfy
    \begin{align}
        \sum\limits_{n=1}^{\infty} \beta_{n}= \frac{C_1}{\epsilon \log{\frac{1}{\epsilon}}}~, \quad 
        \sum\limits_{i=0}^{\infty} \beta_{\theta, n} < C_2,  
    \end{align} 
where $C_1$ and $C_2$ are positive constants. \\ \qed
\end{assumption}
Although not explicitly indicated, the critic learning rate $\beta_n$ is a function of the drift parameter $\varepsilon$, since its infinite sum depends on $\varepsilon$. For notational simplicity, this dependence is omitted throughout the paper.

Traditionally, it is common to choose learning rates such that the corresponding sums diverge, i.e., $\sum_n \beta_n = \infty$, in order to ensure that the learning process does not stagnate. However, in our setting, we instead allow this sum to converge to a finite value that itself diverges as $\varepsilon \to 0$. This choice is motivated by the regime of interest, namely, the rare-transition regime, where $\varepsilon$ is small. In this regime, the sum $\frac{1}{\varepsilon \log(1/\varepsilon)}$ becomes large, thereby ensuring that the effective learning rates remain significant.

This choice is consistent with similar approaches in the social learning literature~\cite{Malek, Malek2}, where algorithmic parameters are designed to account for the latent state dynamics. In particular, the adaptation parameter $\sigma$ in~\eqref{eq:int12} is chosen as a function of $\varepsilon$ to reflect the underlying Markovian drift. Likewise, in our reinforcement learning framework, we impose that the learning rates scale with $\varepsilon$ to ensure proper adaptation and tracking of the latent state, thereby justifying Assumption~\ref{assumption:learning rate}.\par 
With regard to the actor learning rate $\beta_{\theta,n}$, the error in the actor parameters depends on the critic parameter errors, as well as on other quantities influenced by state estimation errors. To effectively suppress this error, the actor step size must decay faster than that of the critic. Accordingly, we assume that $\beta_{\theta,n}$ is chosen to be summable, i.e., $\sum_n \beta_{\theta,n} < \infty$, and independent of $\varepsilon$, to prevent divergence as $\varepsilon \to 0$. Furthermore, we assume that the constant $C_2$ is finite but sufficiently large to avoid premature vanishing of the policy gradient updates.

\begin{assumption}[\textbf{Importance sampling ratio}]
\label{assumption:policies}
We assume that the function $h$, which approximates the individual importance sampling ratios $\rho_{k,n}$ and is defined in \eqref{eq:ind_rho}, satisfies the following Lipschitz property. For any two belief vectors $\mu_1$ and $\mu_2$:
    \begin{align}
\|h(\mu_1)- h(\mu_2)\|\leq B \|\mu_1-\mu_2\| 
    \end{align}
    for some finite $B>0$. 
\end{assumption}
\begin{assumption}[\textbf{Policy function}]
\label{assumption:policy gradient}
We assume that, for any two belief vectors \( \mu_1 \) and \( \mu_2 \)  and $\forall 
 a$ and $a'\in \mathcal{A}$, the policy function $g= g(a'|\mu; \{\theta^{\bar{a}}\}_{\bar{a}\in\mathcal{A}})$ satisfies the following two properties:
    \begin{align}
\|\nabla_{\theta^{a}}\log g (\cdot|
\mu_1)- \nabla_{\theta^a}\log g(\cdot|\mu_2)\|&\leq D \|\mu_1-\mu_2\| \\
    \|\nabla_{\theta^a}\log g \|&\leq E    
    \end{align}
    where $D,E>0$ are constants.\\ \qed
\end{assumption}
 Assumption \ref{assumption:policies} is required in the analysis in order to express the errors in estimating the importance sampling ratio in terms of state estimation errors. Similarly, Assumption \ref{assumption:policy gradient} is introduced to relate policy gradient errors to state estimation errors.  We expect these assumptions to hold for a broad class of policies. For example, the widely used Boltzmann policy model, defined in \eqref{eq:boltzman}, satisfies these assumptions, provided that the arguments of the exponential functions are Lipschitz-continuous. \par 
 \begin{assumption}[\textbf{Combination matrix}]
\label{assumption:combination matrix}
We assume that the combination matrix $C$ satisfies the following property
    \begin{equation}
\rho\left(\mathbb{E}\left[C^T\left(I-\mathds{11}^T / K\right) C\right]\right)< 1
\end{equation}
where $\rho(\cdot)$ denotes the spectral norm of a matrix.\\ \qed 
\end{assumption}
Assumption \ref{assumption:combination matrix} is taken from the MAOPAC setting \cite{SUTTLE20201549}. It is a direct consequence of the strong connectivity assumption on \( G \). This assumption holds if, and only if, the underlying communication graph is connected \cite{Boyd}. \par 

We validate the proposed scheme by comparing performance under two settings: one with full observability, where the true state is known, and another with partial observability, where the true state is unknown and estimated using the social learning scheme described in \eqref{belief_update}--\eqref{eq:int13}. By analyzing the difference between these two scenarios, we aim to characterize the impact of state estimation error on the reinforcement learning component. Our goal is to show that the performance gap between the two settings is both bounded and controllable.

To support this claim, Lemma~\ref{lemma1} first establishes the stability of the learning process by demonstrating that all agents in the network reach consensus in estimating the critic parameters. 

 \begin{lemma}[\textbf{Network stability}]
\label{lemma1}
Let 
\begin{align}
    \mathcal{W}_n & \triangleq \operatorname{col}\left\{\omega_{1, n}, \ldots, \omega_{K,n }\right\} \\ 
    \omega_{c, n}& = \sum_{k=1}^K \omega_{k, n}\\
    \mathcal{W}_{c, n} &\triangleq \mathbbm{1}_K \otimes \omega_{c, n}=\left(\frac{1}{K} \mathbbm{1}_K \mathbbm{1}_K^{T} \otimes I\right) \mathcal { W }_n
\end{align}
Then, we have 
\begin{align}
\|\Delta \mathcal{W}_{c,n}\| \triangleq \|\mathcal{W}_n- \mathcal{W}_{c,n}\|  \text{ converges to } 0 \text{ a.s.}
\end{align}\qed
\end{lemma}
\noindent The proof of Lemma~\ref{lemma1} is presented in Appendix \ref{appendix:lemma1}. The lemma shows that the critic parameters by the individual agents converge almost surely to the same value and this value agrees with their centroid. Given this consensus among agents, we next analyze how closely this centroid approximates the optimal solution. To this end, we define the critic parameter $\omega_{k,n}^{\star}$, which is updated in the same manner as $\omega_{k,n}$ but uses the true belief vectors $\mu_{n}^{\star}$ and $\eta_{n}^{\star}$ instead of the estimated ones, $\mu_{k,n}$ and $\eta_{k,n}$. In Theorem \ref{theorem1} we bound the difference between the network centroid $\mathcal{W}_{c,n}$ and the critic parameter learned under full observability, denoted by $\mathcal{W}_n^{\star} = \operatorname{col}\{\omega_{k,n}^{\star}\}_{k=1}^{K}$.
\begin{theorem}[\textbf{Critic error bound}] 
\label{theorem1}
Under Assumptions \ref{assumption:graph}-\ref{assumption:combination matrix}, and the linear function approximation for state values given in \eqref{eq:linear_approx}, the expected Euclidean distance between the network average $\mathcal{W}_{c,n}$ and the critic parameter under full observation $\mathcal{W}_{n}^{\star}\triangleq \operatorname{col}\{\omega_{k,n}\}_{k=1}^{K}$ is upper bounded by 
\begin{equation}
\label{eq:int115}
    \mathbb{E}\left\|\mathcal{W}_{n+1}^{\star}-\mathcal{W}_{c, n+1}\right\| \leq  \Omega \sum_{i=0}^n \beta_{i} B_{\mu,i} 
\end{equation}
where, for some positive constant $B_{\omega}>0$,
\begin{align}
    \Omega &\triangleq \rho_{\max}^{K-1}BK (1+\gamma) B_{\omega}+\rho_{\max}2(1+\gamma)B_{\omega} \nonumber\\
&\quad+\rho_{\max}^{K-1}BK^2R_{\max}  +7\rho_{\max}KR_{\max}\label{eq:int Phi}
\end{align}
\qed
\end{theorem}
\begin{figure*}[h]
\centering
\begin{tikzpicture}
    \node at (0,0) {
        \begin{minipage}[b]{0.33\linewidth}
          \centering
          \includegraphics[width=5cm]{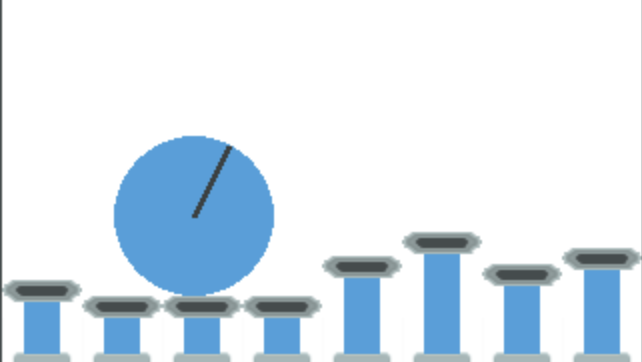}
          \centerline{(a) Phase 3}\medskip
        \end{minipage}
    };
    \node at (5.5,0) {
        \begin{minipage}[b]{0.33\linewidth}
          \centering
          \includegraphics[width=5cm]{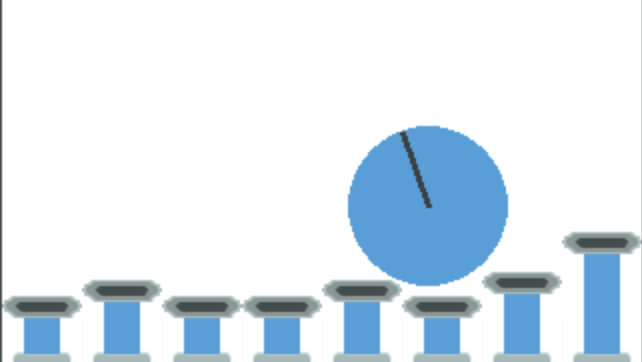}
          \centerline{(b) Phase 2}\medskip
        \end{minipage}
    };
    \node at (11,0) {
        \begin{minipage}[b]{0.33\linewidth}
          \centering
          \includegraphics[width=5cm]{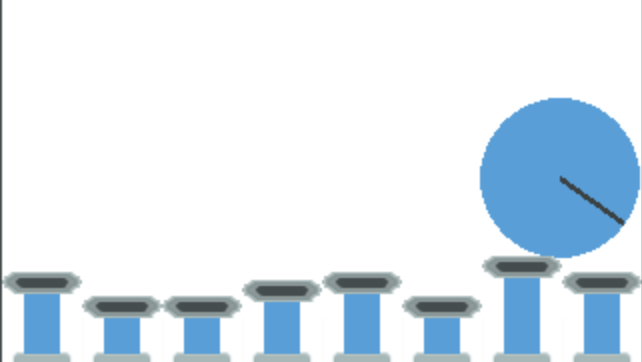}
          \centerline{(c) Phase 1}\medskip
        \end{minipage}
    };

    \draw[-{Triangle[open]}, line width=1mm, draw=black] (6.5,1.8) -- (4.5,1.8);
\end{tikzpicture}

\caption{Illustration of the target ball dynamics in the piston-ball environment.}
\label{fig:piston}
\end{figure*}
The proof of Theorem \ref{theorem1} is given in Appendix \ref{appendix:theorem1}. The bound in \eqref{eq:int115} is expressed as a summation over a decaying term, due to the diminishing step size $\beta_i$. This summation can be further controlled by the state estimation error bound $B_{\mu,i}$, which depends on both the drift parameter $\varepsilon$ and the adaptation parameter $\sigma$. Leveraging these properties, Lemma \ref{lemma:finiteness} establishes that the summation in \eqref{eq:int115} is finite with the proper choice of the adaptation parameter $\sigma$ as in \eqref{eq:int12} and small enough $\varepsilon$. Consequently, the critic values computed by the proposed scheme remain close to those obtained under full state observability, provided that the underlying states evolve slowly. \par 
The assumption of slow state dynamics is adequate, especially in scenarios involving partial observations. In such contexts, a sufficiently stable state over time is necessary to allow the learning algorithm to accumulate enough information for accurate parameter estimation. Importantly, this assumption does not conflict with core principles of reinforcement learning. The requirement that all states be visited infinitely often—a fundamental condition in reinforcement learning—can still be met even if state transitions occur slowly. \par 
As described in Section \ref{sec:algorithm}, the objective of actor-critic algorithms is not only to estimate state values but also to learn optimal policy parameters. To this end, we show that the expected distance between the policy parameter \( \theta_{k,n} \), computed using MARL+SL, and the corresponding parameter obtained under the standard MAOPAC algorithm (i.e., with full observability), can be made arbitrarily small under appropriate conditions. 

Specifically, let \( \theta_{k,n}^{\star} \) denote the actor parameter updated using the same procedure as in Algorithm~\ref{Algorithm1}, but with the true beliefs \( \mu_n^{\star} \) and \( \eta_n^{\star} \) replacing the estimated beliefs \( \mu_{k,n} \) and \( \eta_{k,n} \). Theorem~\ref{theorem2} provides an explicit expression for the expected error bound, along with the conditions under which this result holds.

\begin{theorem}[\textbf{Actor error}]
\label{theorem2}
Under Assumptions \ref{assumption:belief_init}-\ref{assumption:combination matrix}, 
the expected norm of the difference 
$\Delta \theta_{k,n}=\theta_{k,n}^{\star}-\theta_{k,n}$ is bounded by
\begin{align}
    \mathbb{E}\|\Delta\theta_{k,n}\|&\leq \Gamma_1 \sum_{i=0}^{n}\beta_{\theta,i} B_{\mu,i} +E\rho_{\max} \sum_{i=0}^{n} \beta_{\theta, i}\Lambda_i \label{eq:theorem2} 
\end{align}
where, for some constant $B_{\omega}^{\star}>0$, 
\begin{align}
    \Gamma_1&\triangleq (\rho_{\max}^{K-1}B K E B_{\delta}^{\star} +D\rho_{\max} B_{\delta}) \\
    B_{\delta}&\triangleq (R_{\max}+(\gamma+1) B_{\omega}) \\
   B_{\delta}^{\star} &\triangleq (R_{\max}+(\gamma+1) B_{\omega}^{\star})  \\
    \Lambda_i &\triangleq  \gamma B_{\omega}^{\star}B_{\mu, i}+(\gamma+1) \Omega \sum_{j=0}^{i}\beta_j B_{\mu,j}  
\end{align}
 \qed
\end{theorem}
\noindent The proof of Theorem \ref{theorem2} can be found in Appendix \ref{appendix:theorem2}.
The first summation term in \eqref{eq:theorem2} is bounded due to the decaying property of the step size \( \beta_{\theta, n} \), and it remains controllable because it depends on the state estimation error \( B_{\mu, i} \), which is influenced by both the adaptation parameter \( \sigma \) and the drift parameter \( \epsilon \). The second term is also bounded, as established in Lemma~\ref{lemma:finiteness}. Hence, this result justifies that, in slowly changing environments, the performance of MARL+SL can be made arbitrarily close to that of MAOPAC under full observability, given a proper choice of the adaptation parameter.
\begin{lemma}[\textbf{Finiteness of the upper bounds}]
\label{lemma:finiteness}
Under Assumptions~\ref{assumption:belief_init}–\ref{assumption:learning rate}, and provided that the adaptation parameter $\sigma$ is chosen as
\begin{align}
\sigma = \frac{\nu}{\log{\frac{1}{\epsilon}}}, \label{eq:int12}
\end{align}
for some constant $\nu$ satisfying $0 < \nu < \Phi_{\mathrm{min}}$, where $\Phi_{\mathrm{min}}$ is the error exponent defined in~\cite{Malek2}, then, for $\epsilon \to 0$,  the following series are finite:
\begin{align}
\sum_{i=0}^\infty \beta_i B_{\mu,i} < \infty\\
\sum_{i=0}^\infty \beta_{\theta, i} \Lambda_{i} < \infty
\end{align} \qed 
\end{lemma}
 This section highlights one of the main contributions of our work. Our objective was not only to develop a method capable of addressing complex multi-agent settings with partial observability, but also to formulate and analyze it within a theory-driven framework. \par 
For the sake of tractability and feasibility of the analysis, we introduced certain assumptions and slightly simplified the general framework. This reflects the well-recognized trade-off in reinforcement learning between developing theoretically grounded methods and pursuing empirically driven solutions.

\section{Simulation results}
\label{sec:experiment}
Complementing the theoretical analysis, we demonstrate the effectiveness of the proposed method through a series of simulation studies. The proposed scheme was evaluated in the \textit{pistonball} environment~\cite{pettingzoo}, where agents are represented as pistons capable of three discrete actions: move down, stay still, or move up. The global state of the environment is determined by the discretized position of the ball. The agents’ collective objective is to push the ball leftward through coordinated actions as shown in Figure \ref{fig:piston}. Each agent has limited observability of the game environment, restricted to its immediate neighborhood.
\begin{figure}[h]
    \centering
    \includegraphics[width=0.48\textwidth]{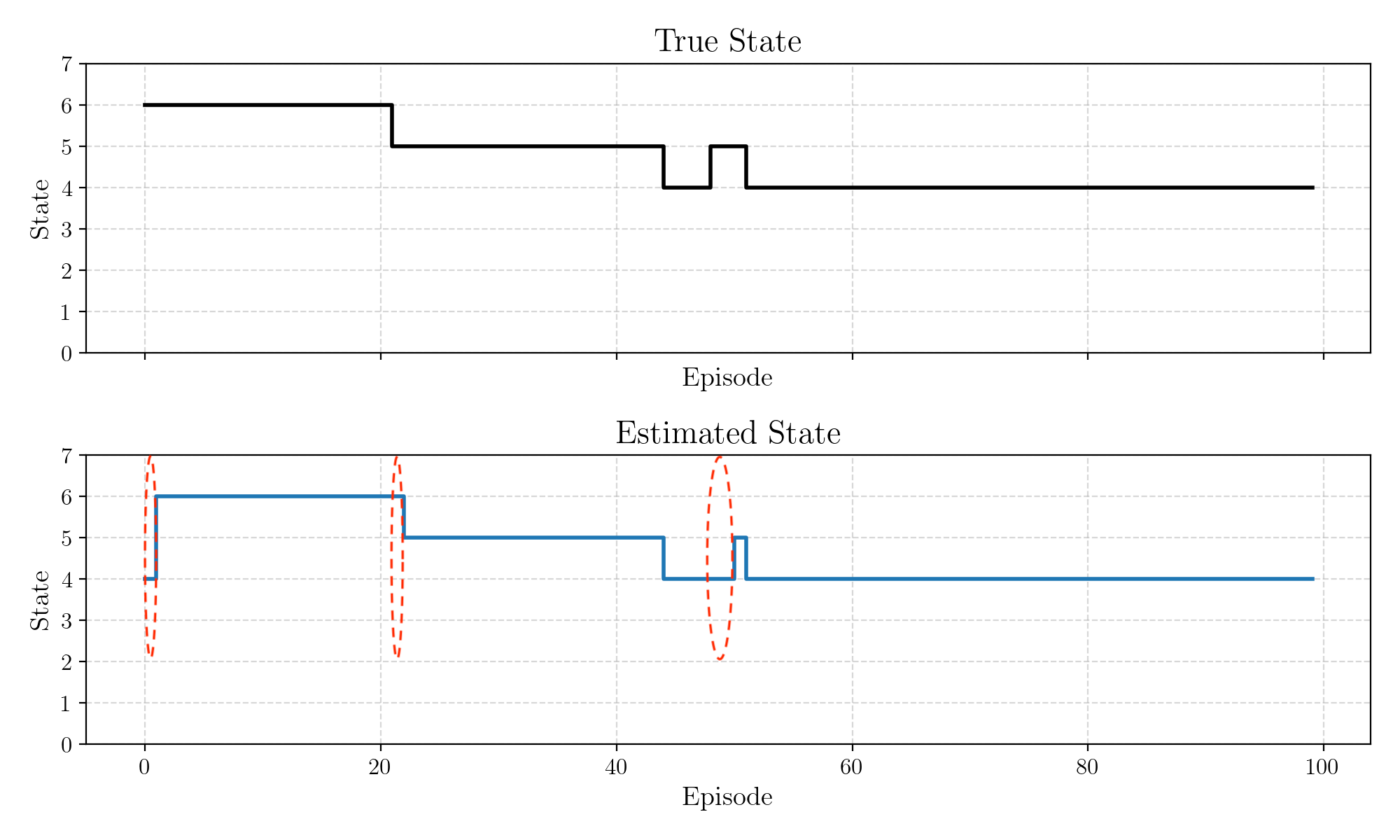}
    \caption{Evolution of the true state and the state estimated via social learning for agent $k = 3$ with state change probability $\epsilon = 0.25$.
}
    \label{fig:belief}
\end{figure}
\par 
To illustrate, assume that the pistons are numbered in ascending order from left to right as $[1, 2, \ldots, K]$.
The \textit{global state} is assigned the value $k$ if the center of the ball is positioned above piston $k$. Consequently, the state space has size $K$. Agents do not observe the exact location of the ball; instead, they receive observation indicating whether the ball is in their vicinity, i.e., being ``seen'' . The variable $\xi_{k,n}$ thus takes one of two possible values, as defined below:
\begin{figure*}[h]
\begin{minipage}[b]{0.5\linewidth}
  \centering
  \centerline{\includegraphics[width=8cm]{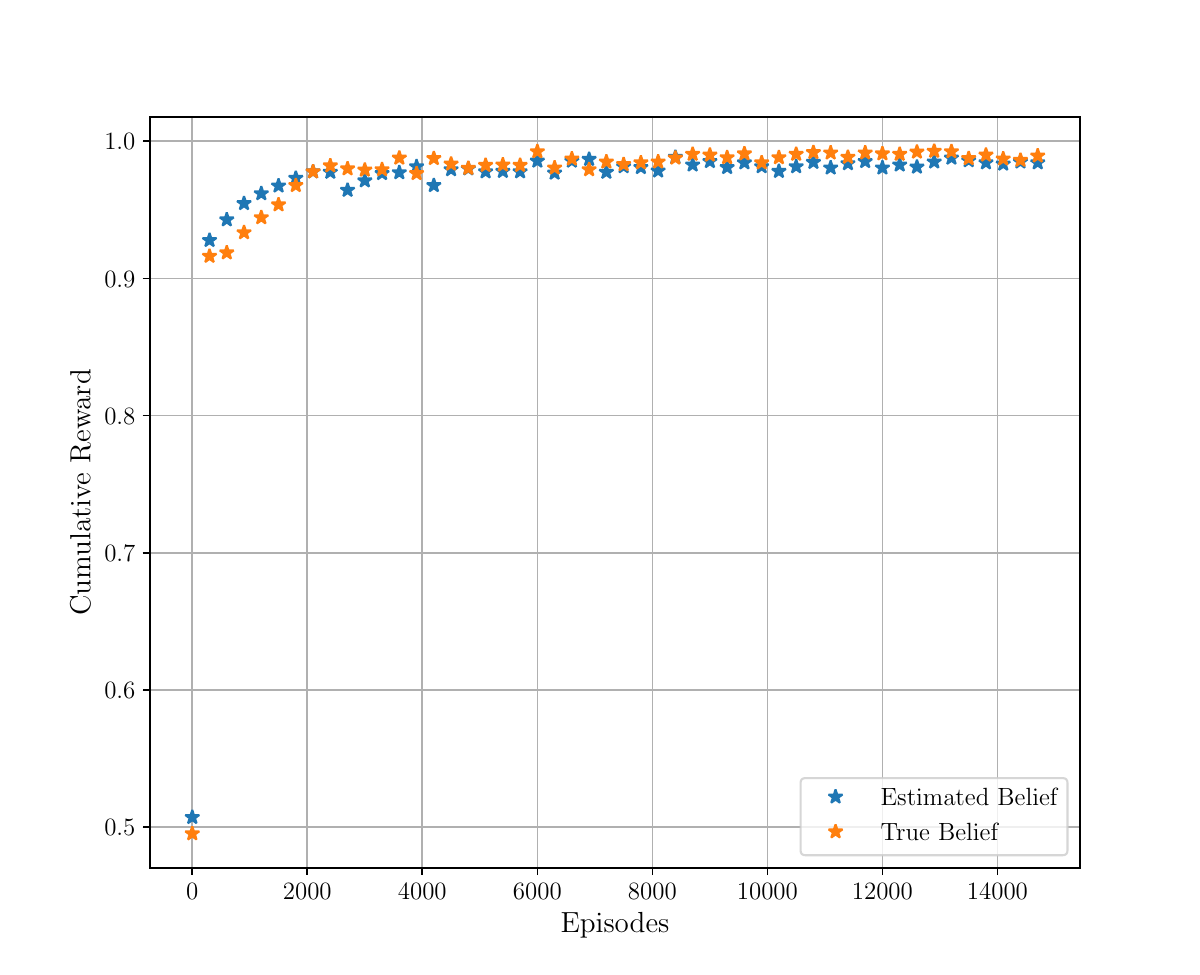}}
\centerline{(a)  $\varepsilon=0.25$ }\medskip
\end{minipage}
\begin{minipage}[b]{0.5\linewidth}
  \centering
  \centerline{\includegraphics[width=8cm]{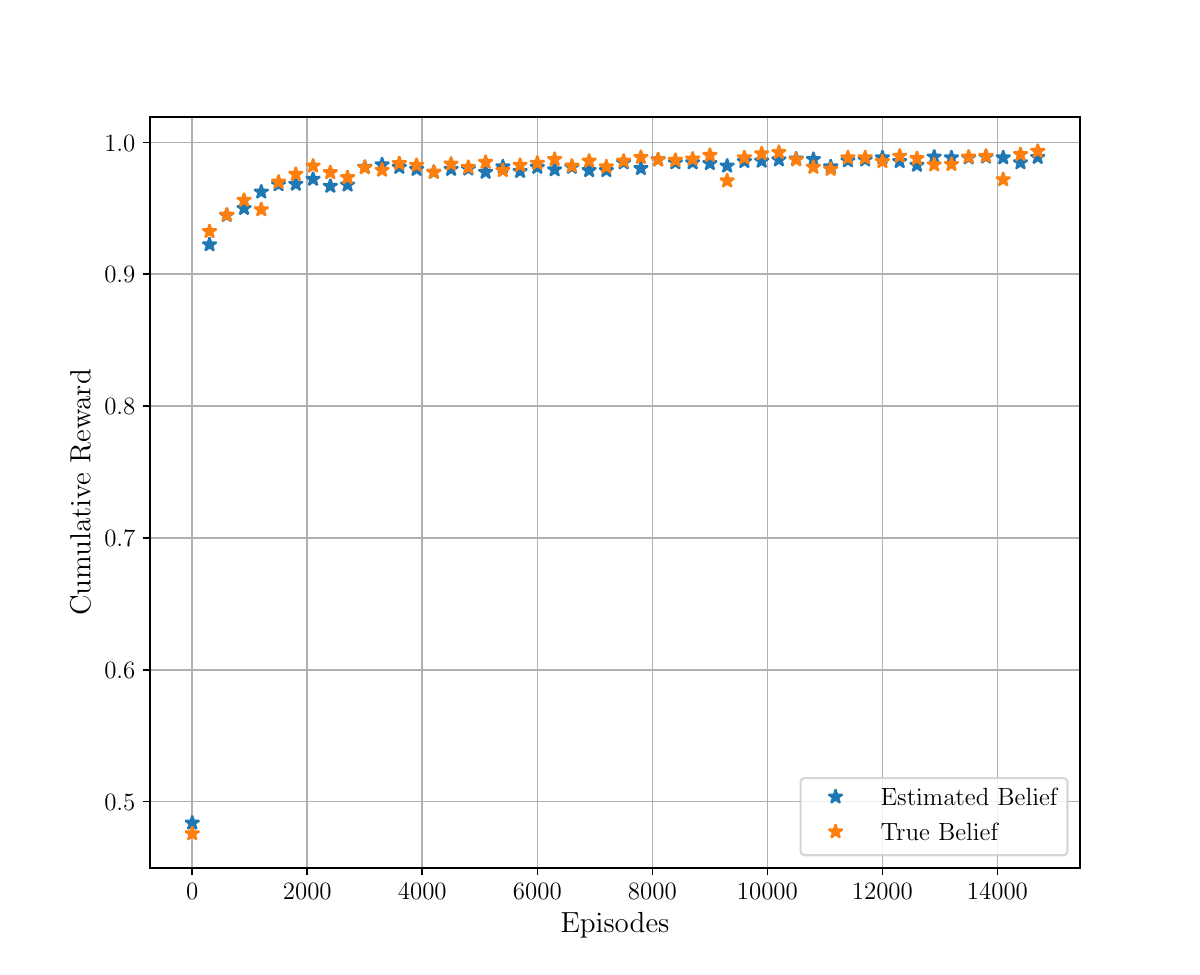}}
  \centerline{(b)  $\varepsilon=0.125$  }\medskip
  
\end{minipage}
\caption{Comparison between the proposed scheme and MAOPAC in terms of cumulative reward for different state change probabilities: $\varepsilon = 0.25$ and $\varepsilon = 0.125$.}
\label{fig:reward}
\end{figure*}
\begin{figure*}[h]
\begin{minipage}[b]{0.5\linewidth}
  \centering
  \centerline{\includegraphics[width=8cm]{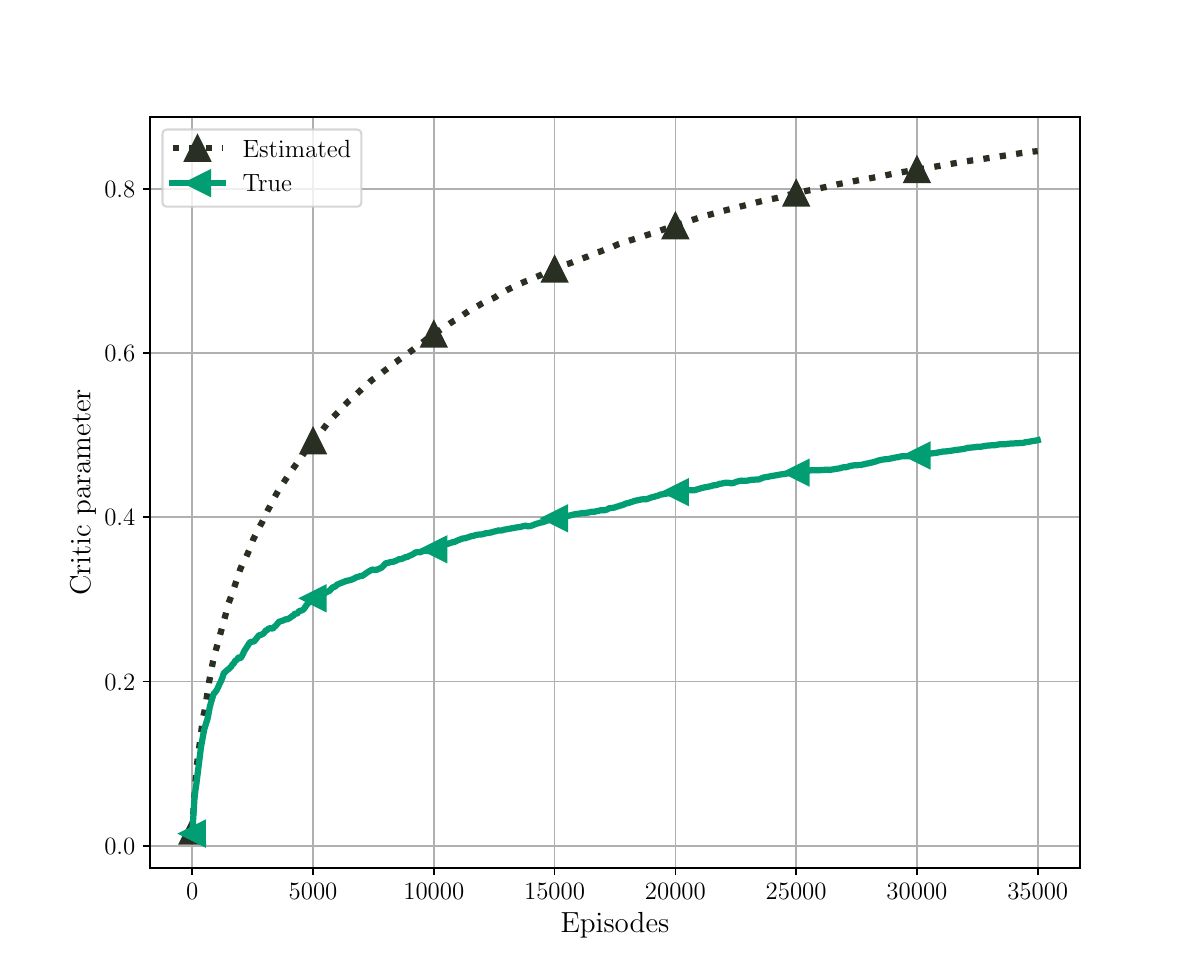}}
  \centerline{(a)  $\varepsilon=0.25$ }\medskip
\end{minipage}
\begin{minipage}[b]{0.5\linewidth}
  \centering
  \centerline{\includegraphics[width=8cm]{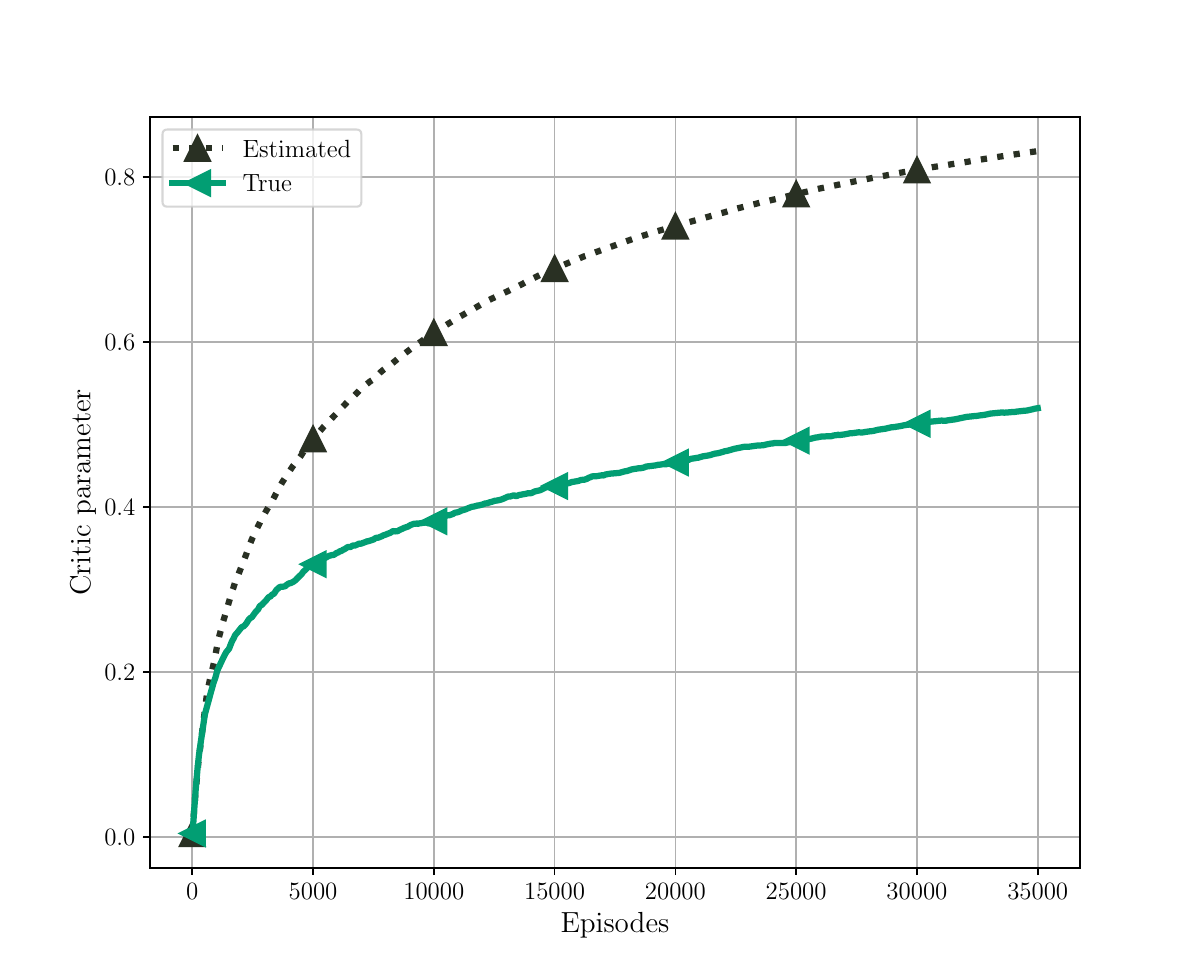}}
  \centerline{(b) $\varepsilon=0.125$ }\medskip

\end{minipage}
\caption{Comparison between the proposed scheme and MAOPAC in terms of critic parameter evolution for state $s = 5$ and agent $k = 3$, under different state change probabilities: $\varepsilon = 0.25$ and $\varepsilon = 0.125$.
}
\label{fig:critic}
\end{figure*}
\begin{figure*}[h]
\begin{minipage}[b]{0.5\linewidth}
  \centering
\centerline{\includegraphics[width=8cm]{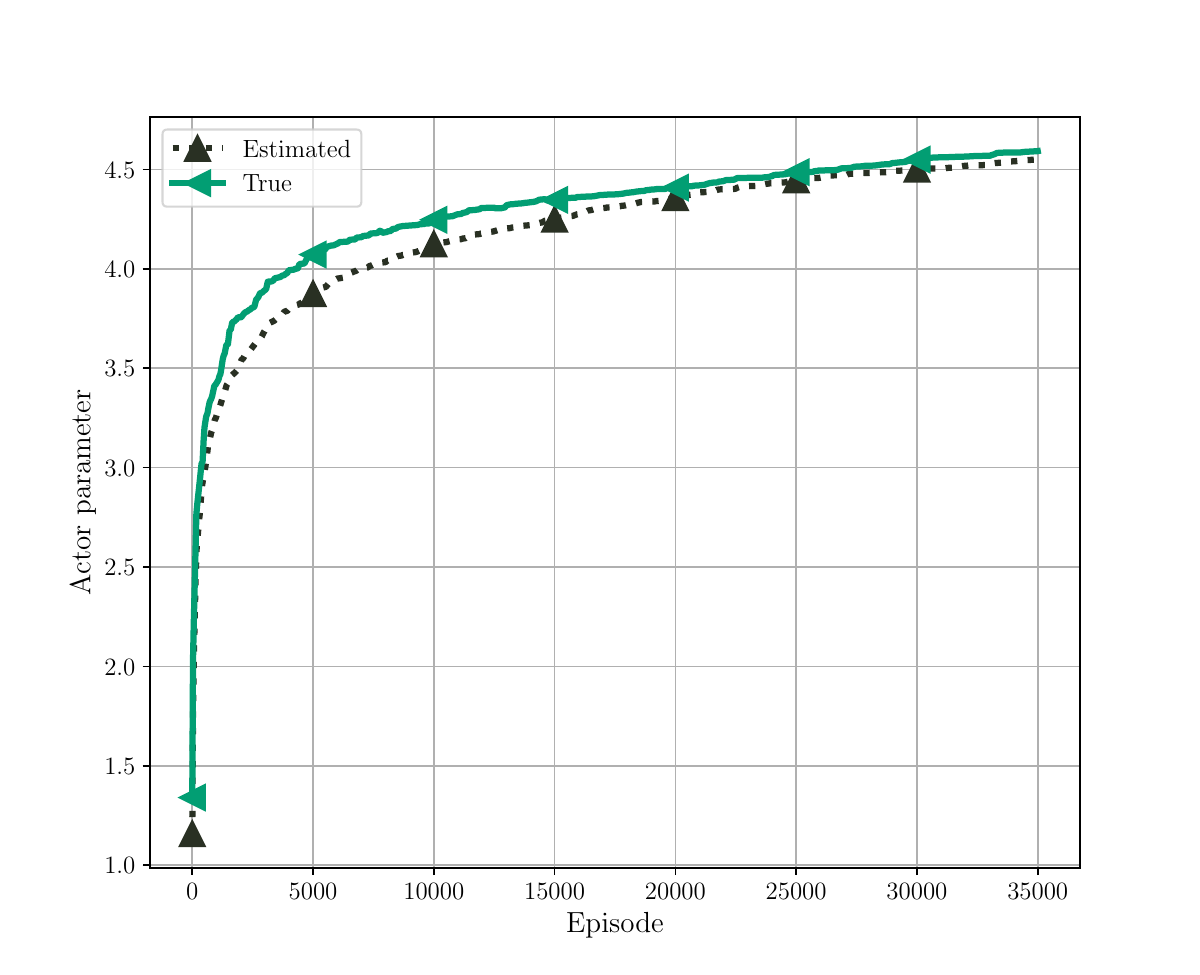}}
  \centerline{(a)  $\varepsilon=0.25$ }\medskip
\end{minipage}
\begin{minipage}[b]{0.5\linewidth}
  \centering
  \centerline{\includegraphics[width=8cm]{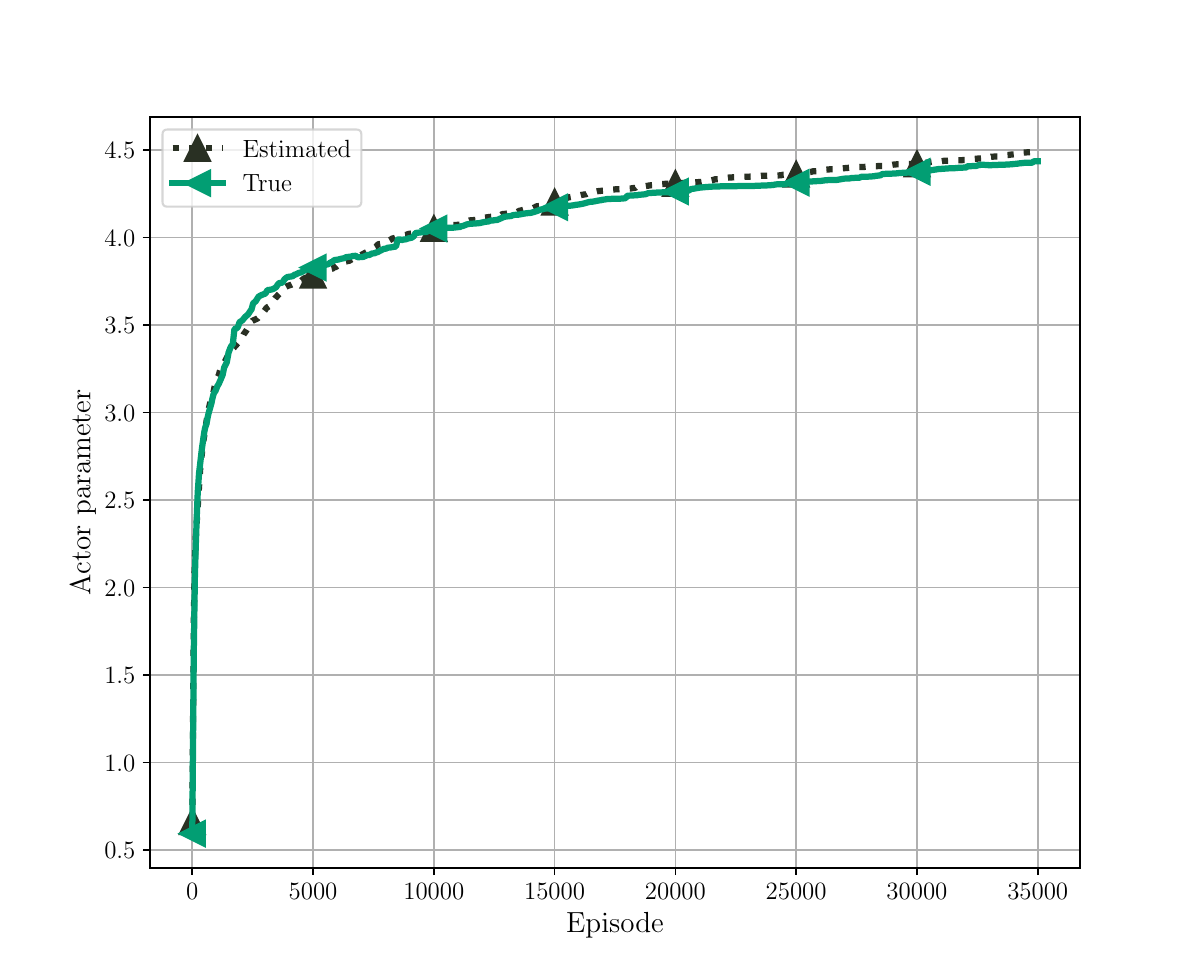}}
  \centerline{(b)  $\varepsilon=0.125$ }\medskip

\end{minipage}
\caption{Comparison between the proposed scheme and MAOPAC in terms of actor parameter evolution for state $s = 5$ and agent $k = 3$, under different drift parameters: $\varepsilon = 0.25$ and $\varepsilon = 0.125$.
 }
\label{fig:actor}
\end{figure*}

\begin{equation}
\xi_{k,n}=
\begin{cases}
\text{``seen'',} & \text{if } s_n \in [k-1,k+1]  \\
\text{``unseen'',} & \text{otherwise}  
\end{cases}
\end{equation}
 We define the likelihood function such that if the ball is completely ``unseen'' by piston $k$, the likelihood becomes uniform over all states, i.e., 
\begin{align}
    \mathcal{L}_k(\xi_{k,n}=\text{``unseen''}|s)= \frac{1}{K}, \forall s\in \mathcal{S}
\end{align}
In other words, when the ball is not visible to the agent---the likelihood function does not provide any new information. \par 
In the case when the ball is ``seen'' by agent~$k$, the corresponding piston assigns higher probabilities to the positions (states) within its local neighborhood and lower probabilities to all other positions, i.e.,
\begin{equation}
  \mathcal{L}_k (\xi_{k,n}=\text{``seen''}|s)=
\begin{cases}
1-\epsilon', & \text{if } s\in [k-1,k+1]  \\
\epsilon', & \text{otherwise}  
\end{cases}
\end{equation}
where  $\epsilon'>0$. The small positive constant $\epsilon'$ is introduced to account for potential noise in the observations.
 \par 
The performance of the proposed scheme is compared with the case where the states are fully observed. In other words, we primarily compare our approach with the existing MAOPAC algorithm. The theoretical analysis shows that the difference between the two approaches can be bounded, which is also confirmed by our simulation results. \par
In particular, we illustrate how the estimated states track the true states. As observed in Figure~\ref{fig:belief}, the estimated state closely matches the true state for the majority of the time. Mismatches in state estimation typically occur at the moments of state transitions. This is because, when a state changes, the belief vector may initially remain biased toward the previous state and requires several iterations (i.e., a transition period) to adjust and align with the new state.
\par
Additionally, we present plots showing the differences in cumulative rewards (Figures~\ref{fig:reward}(a)–(b)), as well as the evolution of the critic (Figures~\ref{fig:critic}(a)–(b)) and actor parameters (Figures~\ref{fig:actor}(a)–(b)) over time for two different state drift parameters: $\varepsilon = 0.25$ and $\varepsilon = 0.125$. These results demonstrate that the performance of the proposed scheme can closely approach that of the MAOPAC algorithm in slowly changing environments. In particular, we observe that the smaller the state change probability, the smaller the difference in the estimated values produced by the two approaches.

\par
Due to the use of hard assignment in state estimation, we observe that the learning parameters are not significantly affected by the magnitude of the difference between the true and estimated beliefs. Instead, the impact primarily depends on how frequently state estimation errors occur.
\par   

 \section{Conclusion}
 \label{sec:conclusion}
In this work, we studied the application of adaptive social learning for state estimation in multi-agent reinforcement learning under partially observable environments. The effectiveness of the proposed MARL+SL approach  was verified both theoretically and experimentally. Under slowly changing environments, the proposed method can achieve performance close to that obtained under full observability. One possible direction for future research is to relax the slow state evolution assumption used in the current setting. Additionally, it would be valuable to explore scenarios in which the likelihood functions are not known a priori but can instead be estimated using appropriate statistical or learning techniques.

\section{Acknowledgment}
We would like to express our gratitude to Mert Kayalıp for motivating and insightful discussions based on his work in \cite{Policy_eval}. We also acknowledge the use of ChatGPT-4.0 for grammar checking and improving the readability of the manuscript.
\appendices
\section{Proof of Lemma \ref{lemma1}}
\label{appendix:lemma1}
To avoid the repeated introduction of new variables, we adopt the convention that variables appearing with a superscript star (`$\star$') in the following proofs correspond to those in Algorithm~1, but are computed using the true belief vector $\mu_n^{\star}$. For example, $\rho_n^{\star}$ denotes the importance sampling ratio computed with knowledge of the true state.
\begin{proof}
Using the update formulas \eqref{alg:omega1} and \eqref{alg:omega2} we can derive 
\begin{align}
&\Delta \mathcal{W}_{c,n}=\mathcal{W}_{n+1}-\mathcal{W}_{c, i+1} =\left(I-\frac{1}{K} \mathbf{1 1}^T\right) \otimes I \nonumber\\
&\left(C_t \otimes I\right)\left(\Delta \mathcal{W}_{c,n}+\beta_{ n} \xi_n\right)
\end{align}
where 
\begin{align}
    \tau_n= \operatorname{col}\{\rho_{k,n}\delta_{k,n}\mu_{k,n}\}_{k=1}^{K}
\end{align}
Therefore, we have 
\begin{align}
&\mathbb{E}\left[\|\Delta \mathcal{W}_{c, n+1}\|^2\vert \mathcal{F}_n\right] \nonumber \\
&\stackrel{(a)}{\leq}\rho  \mathbb{E}\left[\left\| \Delta \mathcal{W}_{c, n}+\beta_{n} \tau_n\right\|^2 \vert \mathcal{F}_n\right] \nonumber \\
& \leq \rho_{\max} \frac{\beta_{n}^2}{\beta_{n+1}^2}\left[\mathbb{E}\left[\left\|\beta_{n}^{-1} \Delta \mathcal{W}_{c,n}\right\|^2 \vert \mathcal{F}_n\right]\right. \nonumber \\
&\left. \quad +2 \mathbb{E}\left[\left\|\beta_{n}^{-1} \tau_{n+1}^T \Delta \mathcal{W}_{c,n}\right\| \vert \mathcal{F}_n\right]+\mathbb{E}\left[\left\| \tau_{n+1}\right\|^2 \vert \mathcal{F}_n\right]\right] \nonumber\\ 
& \leq \rho_{\max} \frac{\beta_{n}^2}{\beta_{n+1}^2}\left[\left\|\beta_{n}^{-1} \Delta \mathcal{W}_{c,n}\right\|^2+2\left\|\beta_{n}^{-1} \Delta \mathcal{W}_{c,n}\right\|\right. \nonumber \\
&\left.\mathbb{E}\left[\left\|\tau_{n+1}\right\|^2 \vert \mathcal{F}_n\right]^{1 / 2}+\mathbb{E}\left[\left\|\tau_{n+1}\right\|^2 \vert \mathcal{F}_n\right]\right] \label{eq:int8}
\end{align}
where $(a)$ is due to Assumption \ref{assumption:combination matrix}.  Next, we bound the term 
\begin{align}
&\left\|\tau_{n+1}\right\|^2 \nonumber\\
& =\sum_{k \in \mathcal{K}}\left\|\rho_{n}\left(r_{k,n}+\gamma \eta_{k,n}^T \omega_{k,n}-\mu_{k,n}^T \omega_{k,n}\right) \mu_{k,n}\right\|^2 \nonumber\\
& \leq \rho_{\max}\sum_{k \in \mathcal{K}}\left\|r_{k,n} +\left(\gamma \eta_{k,n}^T \omega_{k,n}-\mu_{k,n}^T\right) \omega_{k,n} \right\|^2 \nonumber\\
& \leq \rho_{\max}\sum_{k \in \mathcal{K}}\left(|r_{k,n} |^2+2|r_{k,n}| \cdot\left\|\left(\gamma \eta_{k,n}^T-\phi_t^T\right) \omega_{k,n} \right\| \right.\nonumber\\
&\left.\quad +\left\|\left(\gamma \eta_{k,n}^T-\mu_{k,n}^T\right) \omega_{k,n} \right\|^2\right) \nonumber\\
& =\rho_{\max}\sum_{ k\in \mathcal{K}}\left(\left|r_{k,n}\right|^2+2\left|r_{k,n}\right|  \cdot\left|\left(\gamma \eta_{k,n}^T-\mu_{k,n}^T\right) \omega_{k,n}\right| \right. \nonumber\\
&\left.\quad +\left\|\left(\gamma \eta_{k,n}^T-\mu_{k,n}^T\right) \omega_{k,n}\right\|^2\right)
\end{align}
Since rewards and belief vectors are bounded, then for any $M>0$ there exist $K_1>0$ such that 
\begin{equation}
\label{eq:int71}
\mathbb{E}\left[\left\|\tau_{n+1}\right\|^2 \mathbb{I}_{\left\{\sup _{t \leq n}\left\|\mathcal{W}_{t}\right\| \leq M\right\}} \vert \mathcal{F}_n\right] \leq K_1
\end{equation}
Substituting \eqref{eq:int71} into \eqref{eq:int8} we get
\begin{align}
&\mathbb{E}\left[ \|\Delta \mathcal{W}_{c, n+1}\|^2\vert \mathcal{F}_n\right] \nonumber \\
&\leq \rho \mathbb{E}\left[\left\|\Delta \mathcal{W}_{c, n}\right\|^2+2\sqrt{K_1} \cdot\left\|\beta_{n}\Delta \mathcal{W}_{c, n}\right\| \right. \nonumber \\
&\left. \quad + K_1\beta_{n}^2\right]\mathbb{I}_{\left\{\sup _{\tau \leq n}\left\|\mathcal{W}_\tau\right\| \leq M\right\}} \label{eq:int9}
\end{align}
Next dividing both sides of the inequality in \eqref{eq:int9} by $\beta_{n+1}$ we obtain for some $K_2>0$
\begin{align}
&\mathbb{E}\left[\|\beta_{n+1}^{-1} \Delta \mathcal{W}_{c, n+1}\|^2\vert \mathcal{F}_n\right] \leq \frac{\rho \beta_{n}^2}{\beta_{n+1}^2} \mathbb{E}\left[\left\|\Delta \mathcal{W}_{c, n}\beta_n^{-1}\right\|^2 \right.\nonumber\\
&\left.\quad +2\sqrt{K_1} \cdot\left\|\beta_{n}^{-1}\Delta \mathcal{W}_{c, n}\right\| + K_1 \right]\mathbb{I}_{\left\{\sup _{t \leq n}\left\|\mathcal{W}_t\right\| \leq M\right\}} \nonumber\\
&\leq \frac{\rho \beta_{n}^2}{\beta_{n+1}^2} \mathbb{E}\left[\left\|\Delta \mathcal{W}_{c, n}\beta_n^{-1}\right\|^2 \right.\nonumber\\
&\left.\quad+K_2\right]\mathbb{I}_{\left\{\sup _{t \leq n}\left\|\Delta \mathcal{W}_t\right\| \leq M_1\right\}}
\label{eq:int9}
\end{align}
There exist $n_0$ and $\alpha_0$ such that $\forall n>n_0$
\begin{align}
    \frac{\rho \beta_n}{\beta_{n+1}}\leq \alpha_0<1   
\end{align}
To proceed with the analysis, we introduce the following notation:
\begin{equation}
   x_n\triangleq \mathbb{E}\left[ \beta_{n}^{-1}|\|\Delta \mathcal{W}_{c, n}\|^2\vert \mathcal{F}_{n-1}\right] \mathbb{I}_{\left\{\sup _{t \leq n}\left\|\Delta \mathcal{W}_t\right\| \leq M_1\right\}},  
\end{equation}
 then the recursion in \eqref{eq:int9} can rewritten as 
\begin{align}
x_n  &\leq(\alpha_0)^{n-n_0} x_{n_0}+K_2\sum_{i=0}^{n-n_0} \alpha_0^{i} \nonumber\\
&\leq (\alpha_0)^{n-n_0}x_{n_0}+K_2\sum_{i=0}^{\infty} \alpha_0^{i} \nonumber\\
&=(\alpha_0)^{n-n_0} x_{n_0}+\frac{K_2}{1-\alpha_0}\triangleq K<\infty
\end{align}
Therefore we have that 
\begin{align}
    \mathbb{E}\left[\|\Delta \mathcal{W}_{c, n}\|^2\vert \mathcal{F}_{n-1}\right] \mathbb{I}_{\left\{\sup _{\tau \leq n}\left\|\Delta \mathcal{W}_\tau\right\| \leq M_1\right\}} \leq K \beta_{n}^2 \label{eq:int_int}
\end{align}
Summing both sides of the inequality in \eqref{eq:int_int} over all values of $n$ and using Assumption \ref{assumption:learning rate} we get 
\begin{align}
    \sum_{n=0}^{\infty}\mathbb{E}\left[\|\Delta \mathcal{W}_{c, n}\|^2\vert \mathcal{F}_{n-1}\right] \mathbb{I}_{\left\{\sup _{t \leq n}\left\|\mathcal{W}_t\right\| \leq M_1\right\}} < \infty
\end{align}
which implies  
\begin{align}
    \mathbb{E}\left[\sum_{n=0}^{\infty}\|\Delta \mathcal{W}_{c, n}\|^2\mathbb{I}_{\left\{\sup _{t \leq n}\left\|\mathcal{W}_t\right\| \leq M_1\right\}}  \right] < \infty \label{eq:int999}
\end{align}
In order to satisfy the condition in~\eqref{eq:int999}, the following must hold:
\begin{align}
    \sum_{n=0}^{\infty}\|\Delta \mathcal{W}_{c, n}\|^2\mathbb{I}_{\left\{\sup _{t \leq n}\left\|\mathcal{W}_t\right\| \leq M_1\right\}}  < \infty, \text{ a.s.} \label{eq:1234}
\end{align}
Otherwise, if there was a positive probability that the sum in \eqref{eq:1234} diverges, then the expectation in \eqref{eq:int999} would also diverge.
The condition in \eqref{eq:1234}, in turn, implies
\begin{align}
\lim \limits_{n\to \infty}\Delta \mathcal{W}_{c,n} \to  0 \text{ a.s.}
\end{align}
\end{proof}
\section{Proof of Lemma~\ref{lemma:finiteness}}
\label{appendix:lemma2}
\begin{proof}
\noindent \noindent According to~\cite{Malek2}, the state estimation error can be asymptotically bounded as 
\begin{align}
    \limsup_{i \rightarrow \infty} B_{\mu,i} \leq \kappa~ \varepsilon \log{\frac{1}{\varepsilon}} + o\left(\varepsilon \log{\frac{1}{\varepsilon}}\right), 
\end{align}
Therefore, for any $\epsilon$ small, there exists $i(\epsilon) \in \mathbb{N}$ such that $\forall i \geq i(\varepsilon)$
\begin{equation}
    B_{\mu,i} \leq \kappa~ \varepsilon \log{\frac{1}{\varepsilon}} + r(\epsilon), 
\end{equation}
where $r(\epsilon)$ is a positive function verifying
\begin{align}
    \lim_{\epsilon \rightarrow 0 } \frac{r(\epsilon)}{\varepsilon \log{\frac{1}{\varepsilon}}} = 0.
\end{align}
On the other hand, we have 
\begin{align}
\sum_{i=0}^n \beta_i B_{\mu,i} &= \sum_{i=0}^{i(\epsilon)-1} \beta_i B_{\mu,i} + \sum_{i=i(\epsilon)}^{n} \beta_i B_{\mu,i} \nonumber\\
&\leq \sum_{i=0}^{i(\epsilon)-1} \beta_i B_{\mu,i} \nonumber\\
&\quad+ \sum_{i=i(\epsilon)}^{n} \beta_i \left(\kappa~ \varepsilon \log{\frac{1}{\varepsilon}} + r(\epsilon)\right) 
\end{align} 
By taking the limit as $n \rightarrow \infty$:
\begin{align}
    \sum_{i=0}^\infty \beta_i B_{\mu,i} &\leq \sum_{i=0}^{i(\epsilon)-1} \beta_i B_{\mu,i} \nonumber \\
    &\quad + \frac{C_1}{\varepsilon \log{\frac{1}{\varepsilon}}}\left(\kappa~ \varepsilon \log{\frac{1}{\varepsilon}} + r(\epsilon)\right) \nonumber\\
    &= \sum_{i=0}^{i(\epsilon)-1} \beta_i B_{\mu,i} \nonumber \\
    &\quad + C_1 \kappa + \frac{C_1 r(\epsilon)}{\varepsilon \log{\frac{1}{\varepsilon}}} < \infty \label{eq:intint2}
\end{align}
Now, let us consider the term 
\begin{align}
   \Lambda_i &\triangleq  \gamma B_{\omega}^{\star}B_{\mu, i}+(\gamma+1) \Omega \sum_{j=0}^{i}\beta_j B_{\mu,j}   \label{eq:intint}
\end{align}
The summation term in~\eqref{eq:intint} is finite for any \( i \), as shown in~\eqref{eq:intint2}. Therefore, there exists a constant \( C_3 > 0 \) such that \( \Lambda_i \leq C_3 \) for all \( i \). Hence, using the summability of the learning rate \( \beta_{\theta,i} \), we have
\begin{align}
    \sum_{i=0}^{\infty} \beta_{\theta, i} \Lambda_i \leq C_3  \sum_{i=0}^{\infty} \beta_{\theta, i} < \infty 
\end{align}
\end{proof}
\begin{lemma}[\textbf{Bounded critic}]
\label{lemma:bounded critic}
The critic parameters $\|\omega_{k,n}\|$ are bounded almost surely when the states change sufficiently slowly (i.e. $\varepsilon\to 0$) and there exists some constant $B_{\omega}>0$ such that
\begin{align}
    \mathbb{E}(\|\omega_{k,n}\|)\leq B_{\omega}
\end{align}
\end{lemma}
\begin{proof}
To show the boundness of $\omega_{k,n}$, we apply Theorem A.2. from \cite{Actor_critic2}. Let  
\begin{align}
\bar{h}_{k,n}&\triangleq \mathbb{E}\left(\delta_{k,n}^{\star} \mu_{n}^{\star} \vert \mathcal{F}_{n}\right)\\
M_{k,n}&\triangleq\delta_{k,n} \mu_{k,n}-\mathbb{E}\left(\delta_{k,n}^{\star} \mu_{n}^{\star} \vert \mathcal{F}_{n}\right)  \label{eq:int3}
\end{align}
where $\mathcal{F}_{n}$ denotes the history of transitions and rewards up to time $n$, experienced by all agents. \par 
The variable $\bar{h}_{k,n}$ is based on the true state values and hence is defined in the same way as in \cite{Actor_critic2}. Therefore, it automatically satisfies the assumptions of Theorem A.2 \cite{Actor_critic2}. \par 
Now, we need to verify that $M_{k,n}$ also satisfies these conditions. In particular, we need to show that $M_{k,n}$ is a martingale difference and 
\begin{align}
\label{eq:int6}
\mathbb{E}\left(\|M_{k,n}\|^2 \vert \mathcal{F}_n\right)\leq K (1+\|\omega_{k,n}\|^2)
\end{align}
holds for some $K>0$.\par 
\noindent Using \eqref{eq:int3},  we can derive the following 
\begin{align}
    \mathbb{E} (M_{k,n}\vert \mathcal{F}_{n})&= \mathbb{E} \left(\delta_{k,n} \mu_{k,n}-\delta_{n}^{\star} \mu_{n}^{\star} \vert \mathcal{F}_{n}\right) \nonumber\\
    &=\mathbb{E}\left(\Delta \delta_{k,n} \mu_{k,n}+\delta_{k,n}^{\star} \Delta \mu_{n}\vert \mathcal{F}_{n}\right) \label{eq:int5}
\end{align}
Next, we extend the definition of $\Delta \delta_{k,n}$ 
\begin{align}
\Delta \delta_{k,n} &= \delta_{k,n}^{\star}-\delta_{k,n}=\gamma \eta^{\star T}_{n}\omega_{k,n}^{\star}- \mu^{\star T}_{n}\omega_{k,n}^{\star} \nonumber\\
&\quad -(\gamma \eta^{T}_{k,n}\omega_{k,n}-\mu^{T}_{k,n}\omega_{k,n})  \nonumber \\
&= \gamma \Delta \eta^{T}\omega_{k,n}^{\star} +\gamma \eta^{T}\Delta \omega_{k,n}  \nonumber\\
&\quad - \Delta \mu^{T}\omega_{k,n}^{\star} -\mu^{T}\Delta\omega_{k,n} \label{eq:int4}
\end{align}
Substituting  \eqref{eq:int4} into \eqref{eq:int5}, we can get that  $\mathbb{E} (M_{k,n}\vert \mathcal{F}_{n})$ is summation of the terms that is controllable by the state estimation errors. 

Using the fact that  state estimation error decreases and tends to zero when the probability of state change tends to zero \cite{Malek},  we get that, after sufficient number of iterations $N>n$,  $\lim \limits_{\epsilon \to 0 }\mathbb{E} (M_{k,n}\vert \mathcal{F}_n)$ can be approximated as a martingale difference. \\
Next, we check for the condition in \eqref{eq:int6}:
\begin{align}
 &\mathbb{E}\left(\| \delta_{k,n} \mu_{k,n}-\mathbb{E}(\delta_{k,n}^{\star}  \mu_{n}^{\star}|\mathcal{F}_n)\|^2\vert \mathcal{F}_{n})\right) \nonumber\\
&\leq 2 \mathbb{E} \left(\|\delta_{k,n}\mu_{k,n}\|^2 \vert \mathcal{F}_n\right)+2 \|\mathbb{E} \left(\delta_{n}^{\star}\mu_{n}^{\star} \vert \mathcal{F}_n\right)\|^2 \label{eq:int7}
\end{align}
Since rewards and belief vectors are bounded, there exists $k_1$ such that
\begin{align}
\label{eq:int1}
\mathbb{E} \left(\|\delta_{k,n}\mu_{k,n}\|^2 \vert \mathcal{F}_n\right)\leq k_1(1+\|\omega_{k,n}\|^2)
\end{align}
Similarly, there exist $k_2$ and  $k_3$ such that 
\begin{align}
  \mathbb{E} \left(\|\delta_{k,n}^{\star}\mu_{n}^{\star}\|^2 \vert \mathcal{F}_n\right)&\leq k_2(1+\|\omega_{n}^{\star}\|^2) \nonumber\\
  &=k_2(1+\|\omega_{k,n}-\Delta \omega_{k,n}\|^2) \nonumber \\
    &=k_3(1+\|\omega_{k,n}\|^2) \label{eq:int2}
\end{align}
Combining \eqref{eq:int7} with \eqref{eq:int1}-\eqref{eq:int2}  we get 
\begin{align}
      \mathbb{E}(\|\delta_{k,n}\mu_{k,n}\|^2 \vert \mathcal{F}_n)\leq K(1+\|\omega_{k,n}\|^2)
\end{align}
Therefore, we conclude that $\|\omega_{k,n}\|< \infty$ almost surely. \\
This implies that there exists a constant $B_{\omega} > 0$ such that
\begin{align}
    \mathbb{P}(\|\omega_{k,n}\| < B_{\omega}) = 1. \label{eq:int_int222}
\end{align}
Next, we define the following two complementary events:
\begin{align}
    \mathcal{E} \triangleq \left\{\|\omega_{k,n}\| < B_{\omega} \right\}, \quad \mathcal{E}^c \triangleq \left\{\|\omega_{k,n}\| \geq B_{\omega} \right\}.
\end{align}
Hence, using \eqref{eq:int_int222}, we obtain:
\begin{align}
    \mathbb{E}[\|\omega_{k,n}\|] &= \mathbb{P}(\mathcal{E}) \cdot \mathbb{E}[\|\omega_{k,n}\| \mid \mathcal{E}] \nonumber \\
    &\quad + \mathbb{P}(\mathcal{E}^c) \cdot \mathbb{E}[\|\omega_{k,n}\| \mid \mathcal{E}^c] \nonumber \\
    &\leq 1 \cdot \mathbb{E}[\|\omega_{k,n}\| \mid \mathcal{E}] + 0 \cdot \mathbb{E}[\|\omega_{k,n}\| \mid \mathcal{E}^c] \nonumber \\
    &\leq B_{\omega} \label{eq:int_int_int111}
\end{align}
\end{proof}
\section{Proof of Theorem \ref{theorem1}}
\label{appendix:theorem1}
\begin{proof} 
We start by introducing the following notation 
\begin{align}
    H_{n}^{\star} &\triangleq \mu_{n}^{\star} \mu_{ n}^{\star T}-\gamma \mu_{n}^{\star} \eta_{n}^{\star T} \\
        H_{k,n} &\triangleq \mu_{k,n} \mu_{ k,n}^{ T}-\gamma \mu_{k,n} \eta_{k,n}^{T} \\
        \Delta H_{k,n}& \triangleq H_{n}^{\star}-H_{k,n}  \\
d_{k, n} &\triangleq r_{k, n} \mu_{k, n} \\
\mathcal{W}_n & \triangleq \operatorname{col}\left\{\omega_{1, n}, \ldots, \omega_{K, }\right\} \\
\mathcal{C} & \triangleq C \otimes I_K \\
\mathcal{H}_n^{\star} & \triangleq I_K \otimes H_n^{\star} \\
\mathcal{D}_n & \triangleq \operatorname{col}\left\{d_{k, n}\right\}_{k=1}^K\\
\mathcal{H}_n &\triangleq \operatorname{diag}\left\{H_{k, n}\right\}_{k=1}^K
\end{align}
Next, we implement the following rearrangement in the expression  for $\mathcal{W}_{c,n+1}$:
\begin{align}
     \mathcal{W}_{c,n+1}&=
\left(\frac{1}{K} \mathbf{1}_K \mathbf{1}_K^{T} \otimes I\right) \mathcal{W}_{n+1} \nonumber\\
&=\left(\frac{1}{K} \mathbf{1}_K \mathbf{1}_K^{T} \otimes I\right)\mathcal{C}^{T}\left(I-\beta_n \rho_n \mathcal{H}_n\right) \mathcal{W}_n \nonumber\\
&=\left(I-\beta_n \rho_n \mathcal{H}_n\right) \mathcal{W}_{c,n}\nonumber\\
&\quad +\beta_n\rho_n \left(\frac{1}{K} \mathbf{1}_K \mathbf{1}_K^{T} \otimes I\right) \mathcal{D}_n \label{eq:W_est}
\end{align}
where 
\begin{align}
    d_{k,n}^{\star} \triangleq \boldsymbol{r}_{k, n}\mu_n^{\star} 
\end{align}
\begin{equation}
\mathcal{D}_n^{\star} \triangleq \operatorname{col}\{d_{k,n}^{\star}\}_{k=1}^{K}, \quad \mathcal{W}_n^{\star}\triangleq \mathbf{1}_K \otimes \boldsymbol{w}_n^{\star} 
\end{equation}
Similarly, for convenience, we rewrite the definition of $\mathcal{W}_{n+1}^{\star}$:
\begin{equation}
\mathcal{W}_{n+1}^{\star}=\left( I-\beta_n \rho_n^{\star} \mathcal{H}_n^{\star}\right) \mathcal{W}_n^{\star}+\beta_n \rho_n^{\star} \mathcal{D}_n^{\star} \label{eq:W_true}
\end{equation}
Using \eqref{eq:W_est} and \eqref{eq:W_true}, we obtain:
\begin{align}
& \Delta \mathcal{W}_{n}\triangleq\mathcal { W }_{n+1}^{\star}-\mathcal{W}_{c, n+1} \nonumber\\
& =\left(I-\beta_n \rho_n^{\star} \mathcal{H}_n^{\star}\right)\left(\mathcal{W}_n^{\star}-\mathcal{W}_{c, n}\right) \nonumber \\
& \quad -\beta_n \Delta \rho_n \mathcal{H}_n^{\star} \mathcal{W}_{c,n}-\beta_n \rho_n \boldsymbol{\Delta} \mathcal{H}_n \mathcal{W}_{c,n} \nonumber\\
&\quad +\beta_n \Delta \rho_n^{\star} \left(\frac{1}{K} \mathbf{1}_K \mathbf{1}_K^{\mathrm{T}} \otimes I\right)\mathcal{D}_n^{\star}\nonumber\\
&\quad +\beta_n \rho_n \left(\frac{1}{K} \mathbf{1}_K \mathbf{1}_K^{\mathrm{T}} \otimes I\right)\Delta \mathcal{D}_n \label{eq:int333}
\end{align}
where 
\begin{align}
    &\Delta \rho_n\triangleq \rho_n^{\star}-\rho_n, \quad \quad \Delta \mathcal{H}_n \triangleq \mathcal{H}_n^{\star}-\mathcal{H}_n  \\
      & \Delta \mathcal{D}_{n+1}^{\star} \triangleq \mathcal{D}_{n+1}^{\star}-\mathcal{D}_{c, n+1} 
\end{align}
Therefore, the expression in \eqref{eq:int333} yields 
\begin{align}
\label{eq:recursion_base}
\left\|\Delta \mathcal{W}_{n+1}^{\star} \right\| & \leq\left\|I-\rho_{\min}  \beta_n \mathcal{H}_n^{\star}\right\| \| \Delta \mathcal{W}_{n+1}^{\star}  \| \nonumber\\
&\quad +\beta_n \Delta \rho_n \left\|\mathcal{H}_n^{\star}\right\|\left\|\mathcal{W}_{c,n}\right\|  \nonumber\\
&\quad +\beta_n \rho_n \left\|\Delta \mathcal{H}_n\right\|\left\|\mathcal{W}_{c,n}\right\| \nonumber\\
&\quad +\beta_n \Delta \rho_n\left\|{\mathcal{D}}_n^{\star}\right\|+\beta_n \rho_n\left\|{\Delta \mathcal{D}}_n\right\|
\end{align}
Using the fact that $H_n^{\star} \triangleq \mu_{n}^{\star}(\mu_{n}^{\star} - \gamma \eta_{n}^{\star})^{T}$
is a rank-1 matrix and thus has at most one non-zero eigenvalue, the \( \ell_2 \)-norm of \( I - \rho_n \beta_n \mathcal{H}_n^{\star} \) is given by

\begin{align}
    \left\|I-\rho_{n}\beta_n \mathcal{H}_n^{\star}\right\|= \max\left \{1,\lambda' \right\}
\label{eq:matrix}
\end{align}
where 
\begin{align}
       \lambda'=1-\beta_n \rho_n(\|\mu_{n}^{\star}\|^2 -\gamma \mu_{n}^{\star T}\eta_{n}^{\star})  
\end{align}
Since the belief vectors are basis vectors and $\gamma < 1$, the following holds:
\begin{align}
 0\leq (\|\mu_n^{\star}\|^2-\gamma \mu_n^{\star T}\eta_{n}^{\star})\leq 1 
\end{align} 
Therefore, to ensure that the right-hand side of \eqref{eq:matrix} is at most $1$, it is sufficient to impose the following condition:
\begin{align}
    \beta_n \leq \frac{1}{\rho_{\max}} \label{eq:int12345}
\end{align}
Since under \eqref{eq:int12345}, we have $\|I - \beta_n \rho_n \mathcal{H}_n^{\star}\| = 1$, the recursion in \eqref{eq:recursion_base} can be upper bounded by
\begin{align}
    \left\|\Delta \mathcal{W}_{n+1}^{\star} \right\| &\leq   \|\Delta \mathcal{W}^{\star}_0\|+ \sum_{i=0}^{n}   \beta_i \Phi_i \stackrel{(a)}{=} \sum_{i=0}^{n}   \beta_i \Phi_i \label{eq:int789}
\end{align}
where in (a) we used the assumption that the critic parameters—both under full and partial observations—start with the same initial values, and where $\Phi_i$ is defined by  
 \begin{align}
     \Phi_i&\triangleq \| \Delta \rho_i \| \left\|\mathcal{H}_i^{\star}\right\|\left\|\mathcal{W}_{c,i}\right\| + \|\rho_i\| \left\|\Delta \mathcal{H}_i\right\|\left\|\mathcal{W}_{c,i}\right\| \nonumber\\
& \quad + \|\Delta \rho_i\|\left\|{\mathcal{D}}_i^{\star}\right\|+ \rho_i\left\|{\Delta \mathcal{D}}_i\right\|  \label{eq:int23}
 \end{align}
The term $\mathbb{E}\|\Delta \rho_i\|$ can be bounded as:
\begin{align}
\mathbb{E}\|\Delta \rho_i\| &= \mathbb{E}\left \|\prod_{k=1}^K\rho_{k,i}^{\star}-\prod_{k=1}^K \rho_{k,i}\right \| \nonumber\\
&=\mathbb{E} \left\|\sum_{k=1}^{K} \Delta \rho_{k,i}\prod\limits_{\ell=1}^{k-1} \rho_{\ell,i}^{\star} \prod_{z=k+1}^{K} \rho_{k,i}\right\| \nonumber \\
&\leq  \rho_{\max}^{K-1}  \sum_{k=1}^{K} \|\Delta \rho_{k,i} \|\stackrel{(a)}{\leq}  \rho_{\max}^{K-1} B \sum_{k=1}^{K}\|\Delta \mu_{k,i} \| \nonumber \\
&=\rho_{\max}^{K-1} B K B_{\mu,i} \label{eq:int_delta_rho}
\end{align}
where $(a)$ is due to Assumption \ref{assumption:policies} and the term   $\Delta \rho_{k,n}$ is defined as 
\begin{align}
\Delta \rho_{k,n}\triangleq \rho_{k,n}^{\star}-\rho_{k,n} 
\end{align}
 Next, we find a bound for $\mathbb{E}\left\|\Delta H_{k,i}\right\|$:
 \begin{align}
 \label{eq:int24}
     \mathbb{E} \left\| \Delta H_{k,i}\right\|&\leq \mathbb{E}\left\|\Delta \mu_{k,i} \mu_{ i}^{\star T}+\mu_{k,i} \Delta \mu_{k,i}^{T}\right\|\nonumber\\
     &\quad +\gamma \mathbb{E}\left\|\Delta\mu_{k,i} \eta_{i}^{\star T}+\mu_{k,i} \Delta \eta_{i}^{T}\right\| \nonumber\\
     &\leq (2+\gamma)\mathbb{E}\|\Delta \mu_{k,i}\|+\gamma \mathbb{E}\|\Delta \eta_{k,i}\|\nonumber\\
        &\leq 2(1+\gamma)B_{\mu,i}
 \end{align}
 We assume that rewards are bounded, i.e. $r_{k,n}\leq R_{\max}$. Therefore, $\mathbb{E}\|\mathcal{D}_i^{\star}\|$ can be bounded by 
 \begin{align}
\label{eq:int21}
    \mathbb{E} \|\mathcal{D}_i^{\star}\| \leq K R_{\max}
 \end{align}
Next, we derive a bound on the term \( \left\| \Delta \mathcal{D}_i \right\| \) as follows:
\begin{align}
\label{eq:int22}
   \mathbb{E} \left\|{\Delta \mathcal{D}}_i\right\| & \leq \sum_{k=1}^{K} \mathbb{E} \|\left(r_{k, i}^{\star}\mu_i^{\star}-r_{k, n}\mu_{k,i}\right) \|\nonumber\\
    &\leq K R_{\max} \mathbb{E} \|\Delta \mu_{k,i}\|+\sum_{k=1}^{K} \mathbb{E} \|\Delta r_{k,i}\|
\end{align}
where $\Delta r_{k,i}\triangleq (r_{k,i}^{\star}-r_{k,i})$.\par 
Noting that \( r_{k,i}^{\star} \triangleq r_{k}(s_i^{\star}, a_i^{\star}, s_{i+1}^{\star}) \) and \( r_{k,i} \triangleq r_{k}(s_i, a_i, s_{i+1}) \), where \( s_i^{\star} \) and \( a_{k,i}^{\star} \) denote the global state and the action taken by agent \( k \) at time \( i \) under full observability, we can deduce that
\begin{align}
    \mathbb{E} \|\Delta r_{k,i}\|
   & \stackrel{(a)}{\leq} \mathbb{E}(\Delta r_{k,i}|s_i^{\star}\neq s_i)\mathbb{P}(s_i^{\star}\neq s_i) \nonumber\\
    &\quad +\mathbb{E}(\Delta r_{k,i}|a_{k,i}^{\star}\neq a_{k,i})\mathbb{P}(a_{k,i}^{\star}\neq a_{k,i})\nonumber\\
    & \quad +\mathbb{E}(\Delta r_{k,i}|s_{i+1}^{\star}\neq s_i)\mathbb{P}(s_{i+1}^{\star}\neq s_{i+1})\nonumber\\
    & \quad 0 * \mathbb{P}(s_i^{\star}= s_i, a_{k,i}^{\star}=a_{k,i}, s_{i+1}^{\star}= s_{i+1})\nonumber\\
    &\stackrel{(b)}{\leq} \mathbb{E}(\Delta r_{k,i}|a_{k,i}^{\star}\neq a_{k,i})\mathbb{P}(a_{k,i}^{\star}\neq a_{k,i})\nonumber\\
    &\quad +2R_{\max}B_{\mu,i} +2R_{\max}B_{\mu,i+1}\label{eq:int2345}
\end{align}
where in (a) we apply the \textit{law of total expectation} and the \textit{union bound}, and in (b) we use the reward and state estimation error bounds.\par 
For simplicity of analysis, we assume that the behavioral policies are deterministic and disregard the randomness they induce, focusing solely on the uncertainty arising from state estimation. Under this assumption, we have \( a_{k,i}^{\star} \neq a_{k,i} \) if and only if the current state is estimated incorrectly, i.e., \( s_i^{\star} \neq s_i \). Hence, \( \mathbb{P}(a_{k,i}^{\star} = a_{k,i}) = \mathbb{P}(s_i^{\star} = s_i) \), and the inequality in~\eqref{eq:int2345} can be rewritten as:
\begin{align}
    \mathbb{E} \|\Delta r_{k,i}\|
    &\leq  \mathbb{E}(\Delta r_{k,i}|a_{k,i}^{\star}\neq a_{k,i})\mathbb{P}(s_{i}^{\star}\neq s_{i})\nonumber\\
    &\quad + 2R_{\max}B_{\mu,i} +2R_{\max}B_{\mu,i+1}\nonumber\\
    &\leq 4R_{\max}B_{\mu,i} +2R_{\max}B_{\mu,i+1} \label{eq:intaaa}
\end{align}
Substituting \eqref{eq:intaaa} into \eqref{eq:int22}, we obtain  
\begin{align}
    \mathbb{E}\|\Delta \mathcal{D}_i \| &\leq 5 K R_{\max}B_{\mu, i}+2 K R_{\max}B_{\mu, i+1} \nonumber \\
    &\stackrel{(a)}{\leq} 7 K R_{\max}B_{\mu, i}\label{eq:intbbb}
\end{align}
where (a) holds because the bound on the expected state estimation error is monotonically decaying, i.e., \( B_{\mu, i+1} \leq B_{\mu, i} \).
 \par  
Next, substituting \eqref{eq:int_delta_rho}-\eqref{eq:int21} and \eqref{eq:intbbb} into \eqref{eq:int23} ,  using Assumption \ref{assumption:policies} and Lemma \ref{lemma:bounded critic}, the variable $\mathbb{E}\Phi_{i}$ can be bounded by 
\begin{align}
\mathbb{E}\Phi_{i} &\leq \rho_{\max}^{K-1}BK (1+\gamma) B_{\omega}B_{\mu,i}\nonumber\\
&\quad +\rho_{\max}2(1+\gamma)B_{\omega} B_{\mu ,i} \nonumber\\
&\quad+\rho_{\max}^{K-1}BK^2R_{\max} B_{\mu,i} \nonumber\\
&\quad +7\rho_{\max}KR_{\max}B_{\mu,i} \label{eq:int Phi}
\end{align}
Using \eqref{eq:int Phi}, the recursion in \eqref{eq:int789} can be written as 
\begin{align}
   \mathbb{E} \left\|\Delta \mathcal{W}_{n+1}^{\star} \right\| &= \Omega \sum_{i=0}^{n} \beta_i B_{\mu, i} 
\end{align}
\end{proof}
\section{Proof of Theorem 2}
\label{appendix:theorem2}
\begin{proof}
We start with the updates for $\theta_{k,n}^a$ and $\theta_{k,n}^{\star a}$
\begin{align} 
\label{eq:int17}
\Delta \theta_{k,n+1}^{a}&=\Delta \theta_{k,n}^{a}+ \beta_{\theta, n}(\rho_n^{\star}\delta_{k,n}^{\star} \Psi_{k,n}^{\star a}- \rho_n\delta_{k,n} \Psi_{k,n}^a) \nonumber\\
&= \Delta \theta_{k,n}^a+ \beta_{\theta, n}\Delta\rho_{k,n}\delta_{k,n}^{\star}\Psi_{k,n}^{\star a} \nonumber\\
&+\beta_{\theta, n}\rho_{k,n}\Delta\delta_{k,n} \Psi_{k,n}^{\star a} \nonumber \\
&+\beta_{\theta, n}\rho_{k,n}\delta_{k,n} \Delta \Psi_{k,n}^{a}
\end{align}
where 
\begin{align}
    \Delta\Psi_{k,n}^a\triangleq \Psi_{k,n}^{\star a}-\Psi_{k,n}^{a}
\end{align}
Assuming $\Delta\theta_{k,0}=0$, the recursion in \eqref{eq:int17} can be rewritten as
\begin{align}
   \Delta \theta_{k,n+1}^a &= \sum_{i=0}^{n}\beta_{\theta, i}\Delta\rho_{k,i}\delta_{k,i}^{\star}\Psi_{k,i}^{\star a } \nonumber\\
   &+\sum_{i=0}^{n}\beta_{\theta, i}\rho_{k,i}\Delta\delta_{k,i} \Psi_{k,i}^{\star a} \nonumber\\
   &+\sum_{i=0}^{n}\beta_{\theta, i}\rho_{k,i}\delta_{k,i} \Delta \Psi_{k,i}^{a} 
\end{align}
Using Assumptions \ref{assumption:policies} and \ref{assumption:policy gradient} we can derive 
\begin{align}
\label{eq:int52}
     \mathbb{E} \|\Delta \theta_{k,n+1}^a\|&\leq E\sum_{i=0}^{n}\beta_{\theta, i}     \mathbb{E}(\|\Delta\rho_{k,i}\|\|\delta_{k,i}^{\star}\|) \nonumber\\
   &\quad+     E\sum_{i=0}^{n}\beta_{\theta, i} \mathbb{E}(\rho_{k,i}\|\Delta\delta_{k,i}\| ) \nonumber \\
&\quad+D\sum_{i=0}^{n}\beta_{\theta, i}\mathbb{E}(\rho_{k,i}\|\delta_{k,i}\|\|\Delta \mu_{k,i}\|) 
\end{align}
In addition, using \eqref{eq:int4}, the difference $\Delta \delta_{k,n}$ can be bounded by 
\begin{align}
\label{eq:int41}
    \|\Delta \delta_{k,i}\| &\leq \|\gamma \Delta \eta^{T}_{k,i}\omega_{k,n}^{\star} +\gamma \eta^{T}_{k,i}\Delta \omega_{k,i} \nonumber\\
&\quad - \Delta \mu^{T}\omega_{k,i}^{\star} -\mu^{T}\Delta\omega_{k,i} \|
\end{align}
The almost sure boundedness of $\|\omega_{k,n}^{\star}\|$ is established in \cite{SUTTLE20201549} (Lemma A.3). Therefore, following arguments similar to those used in the derivation of \eqref{eq:int_int_int111}, there exists a constant $B_{\omega}^{\star} > 0$ such that for all $n$,
\begin{align}
    \mathbb{E}\|\omega_{k,n}^{\star} \|\leq B_{\omega}^{\star}
\end{align}
Now, using the result of Theorem \ref{theorem1} we can rewrite \eqref{eq:int41} as:
\begin{align}
\label{eq:int51}
   \mathbb{E} \|\Delta \delta_{k,i}\| &\leq \gamma B_{\omega}^{\star}B_{\mu, i}+(\gamma+1) \Omega \sum_{j=0}^{i}\beta_j B_{\mu,j}\triangleq \Lambda_i 
\end{align}
Using the bounds for critic parameters  and rewards,  we can show that 
\begin{align}
    \mathbb{E}\|\delta_{k,n}\|&= \mathbb{E}\|r_{k,n}+\gamma\eta_{k,n}^{T}\omega_{k,n}-\mu_{k,n}^{T}\omega_{k,n}\| \nonumber \\
    &\leq R_{\max}+(\gamma+1) B_{\omega} \triangleq B_{\delta} \label{eq:delta_bound}
\end{align}
Similarly, we can derive the bound for $ \mathbb{E}\|\delta_{k,n}^{\star}\|$
\begin{align}
    \mathbb{E}\|\delta_{k,n}^{\star}\| &\leq R_{\max}+(\gamma+1) B_{\omega}^{\star} \triangleq B_{\delta}^{\star}  \label{eq:delta_star_bound}
\end{align}
Substituting  \eqref{eq:int_delta_rho} and \eqref{eq:int51}-\eqref{eq:delta_star_bound} into  \eqref{eq:int52} we get 
\begin{align}
   \mathbb{E} \|\Delta\theta_{k,n}\|&\leq \rho_{\max}^{K-1}B K E B_{\delta}^{\star} \sum_{i=0}^{n}\beta_{i,\theta} B_{\mu,i} \nonumber\\
   &+E\rho_{\max} \sum_{i=0}^{n} \beta_{\theta, i}\Lambda_i \nonumber\\
     &+D\rho_{\max} B_{\delta}\sum_{i=0}^{n}\beta_{\theta,i} B_{\mu,i}
\end{align}
\end{proof}
\bibliographystyle{IEEEtran}
\bibliography{references}

\begin{thebibliography}{10}
\providecommand{\url}[1]{#1}
\csname url@samestyle\endcsname
\providecommand{\newblock}{\relax}
\providecommand{\bibinfo}[2]{#2}
\providecommand{\BIBentrySTDinterwordspacing}{\spaceskip=0pt\relax}
\providecommand{\BIBentryALTinterwordstretchfactor}{4}
\providecommand{\BIBentryALTinterwordspacing}{\spaceskip=\fontdimen2\font plus
\BIBentryALTinterwordstretchfactor\fontdimen3\font minus \fontdimen4\font\relax}
\providecommand{\BIBforeignlanguage}[2]{{%
\expandafter\ifx\csname l@#1\endcsname\relax
\typeout{** WARNING: IEEEtran.bst: No hyphenation pattern has been}%
\typeout{** loaded for the language `#1'. Using the pattern for}%
\typeout{** the default language instead.}%
\else
\language=\csname l@#1\endcsname
\fi
#2}}
\providecommand{\BIBdecl}{\relax}
\BIBdecl

\bibitem{NEXP}
D.~S. Bernstein, R.~Givan, N.~Immerman, and S.~Zilberstein, ``The complexity of decentralized control of {M}arkov decision processes,'' \emph{Mathematics of Operations Research}, vol.~27, no.~4, pp. 819--840, 2002.

\bibitem{Policy_eval}
M.~Kayaalp, F.~Ghadieh, and A.~H. Sayed, ``Policy evaluation in decentralized {POMDP}s with belief sharing,'' \emph{IEEE Open Journal of Control Systems}, vol.~2, pp. 125--145, 2023.

\bibitem{SUTTLE20201549}
W.~Suttle, Z.~Yang, K.~Zhang, Z.~Wang, T.~Başar, and J.~Liu, ``A multi-agent off-policy actor-critic algorithm for distributed reinforcement learning,'' \emph{IFAC-PapersOnLine}, vol.~53, no.~2, pp. 1549--1554, 2020.

\bibitem{Ainur}
\BIBentryALTinterwordspacing
A.~Zhaikhan and A.~H. Sayed, ``Multi-agent off-policy actor-critic reinforcement learning for partially observable environments,'' 2024. [Online]. Available: \url{https://arxiv.org/abs/2407.04974}
\BIBentrySTDinterwordspacing

\bibitem{IL1}
\BIBentryALTinterwordspacing
C.~S.~D. Witt, T.~Gupta, D.~Makoviichuk, V.~Makoviychuk, P.~H.~S. Torr, M.~Sun, and S.~Whiteson, ``Is independent learning all you need in the starcraft multi-agent challenge?'' 2020. [Online]. Available: \url{https://api.semanticscholar.org/CorpusID:227054146}
\BIBentrySTDinterwordspacing

\bibitem{QMIX}
T.~Rashid, M.~Samvelyan, C.~S. De~Witt, G.~Farquhar, J.~Foerster, and S.~Whiteson, ``Monotonic value function factorisation for deep multi-agent reinforcement learning,'' \emph{J. Mach. Learn. Res.}, vol.~21, no.~1, Jan. 2020.

\bibitem{COMA}
J.~N. Foerster, G.~Farquhar, T.~Afouras, N.~Nardelli, and S.~Whiteson, ``Counterfactual multi-agent policy gradients,'' in \emph{Proc. of Conference on Artificial Intelligence}, 2018, pp. 2974 -- 2982.

\bibitem{MADPPG}
R.~Lowe, Y.~Wu, A.~Tamar, J.~Harb, P.~Abbeel, and I.~Mordatch, ``Multi-agent actor-critic for mixed cooperative-competitive environments,'' in \emph{Advances in International Conference on Neural Information Processing Systems}, 2017, p. 6382–6393.

\bibitem{VDN}
P.~Sunehag, G.~Lever, A.~Gruslys, W.~M. Czarnecki, V.~Zambaldi, M.~Jaderberg, M.~Lanctot, N.~Sonnerat, J.~Z. Leibo, K.~Tuyls, and T.~Graepel, ``Value-decomposition networks for cooperative multi-agent learning based on team reward,'' in \emph{Proc. International Conference on Autonomous Agents and MultiAgent Systems}, 2018, p. 2085–2087.

\bibitem{QTRAN}
K.~Son, D.~Kim, W.~J. Kang, D.~E. Hostallero, and Y.~Yi, ``{QTRAN}: Learning to factorize with transformation for cooperative multi-agent reinforcement learning,'' in \emph{Proceedings of the 36th International Conference on Machine Learning}, vol.~97, 2019, pp. 5887--5896.

\bibitem{Neural_POMDP}
G.~Stamatelis and N.~Kalouptsidis, ``Active hypothesis testing in unknown environments using recurrent neural networks and model free reinforcement learning,'' in \emph{Proc. European Signal Processing Conference (EUSIPCO)}, 2023, pp. 1380--1384.

\bibitem{Actor_critic_POMDP}
B.~Liu, H.~He, and D.~W. Repperger, ``Two-time-scale online actor-critic paradigm driven by pomdp,'' in \emph{Proc. International Conference on Networking, Sensing and Control (ICNSC)}, 2010, pp. 243--248.

\bibitem{POMDP_RNN}
S.~Duell, S.~Udluft, and V.~Sterzing, ``Solving partially observable reinforcement learning problems with recurrent neural networks,'' in \emph{Neural Networks: Tricks of the Trade: Second Edition}.\hskip 1em plus 0.5em minus 0.4em\relax Springer, 2012, pp. 709--733.

\bibitem{Mean_field}
M.~Yang, G.~Liu, Z.~Zhou, and J.~Wang, ``Partially observable mean field multi-agent reinforcement learning based on graph attention network for uav swarms,'' \emph{Drones}, vol.~7, no.~7, 2023.

\bibitem{Belief1}
B.~Eker, E.~Ozkucur, C.~Mericli, T.~Mericli, and H.~L. Akin, ``A finite horizon dec-{POMDP} approach to multi-robot task learning,'' in \emph{Proc. International Conference on Application of Information and Communication Technologies (AICT)}, 2011, pp. 1--5.

\bibitem{LSTM}
B.~Bakker, ``Reinforcement learning with long short-term memory,'' in \emph{Proc. International Conference on Neural Information Processing Systems: Natural and Synthetic}, 2001, p. 1475–1482.

\bibitem{NN_POMDP}
T.~Phan, F.~Ritz, P.~Altmann, M.~Zorn, J.~N\"{u}\ss{}lein, M.~K\"{o}lle, T.~Gabor, and C.~Linnhoff-Popien, ``Attention-based recurrence for multi-agent reinforcement learning under stochastic partial observability,'' in \emph{Proc. International Conference on Machine Learning}, 2023, p. 27840–27853.

\bibitem{NN_model_free}
N.~Yang, H.~Zhang, and R.~Berry, ``Partially observable multi-agent deep reinforcement learning for cognitive resource management,'' in \emph{IEEE Global Communications Conference}, 2020, pp. 1--6.

\bibitem{NN_model_free2}
Z.~Xu, Y.~Bai, D.~Li, B.~Zhang, and G.~Fan, ``Side: State inference for partially observable cooperative multi-agent reinforcement learning,'' in \emph{Proceedings International Conference on Autonomous Agents and Multiagent Systems}, 2022, p. 1400–1408.

\bibitem{Consensus1}
W.~Zhang, X.~Chen, and L.~Ma, ``Online planning for multi-agent systems with consensus protocol,'' in \emph{Proceedings of the 33rd Chinese Control Conference}, 2014, pp. 1126--1131.

\bibitem{Similar_1}
M.~Peti, F.~Petric, and S.~Bogdan, ``Decentralized coordination of multi-agent systems based on {POMDP}s and consensus for active perception,'' \emph{IEEE Access}, vol.~11, pp. 52\,480--52\,491, 2023.

\bibitem{Consensus2}
W.~Zhang, L.~Ma, and X.~Li, ``Multi-agent reinforcement learning based on local communication,'' \emph{Cluster Computing}, vol.~22, no.~6, pp. 15\,357--15\,366, 2019.

\bibitem{MA_POMDP}
Y.~Zhang and M.~M. Zavlanos, ``Cooperative multiagent reinforcement learning with partial observations,'' \emph{IEEE Transactions on Automatic Control}, vol.~69, no.~2, pp. 968--981, 2024.

\bibitem{REINFORCE}
R.~J. Williams, ``Simple statistical gradient-following algorithms for connectionist reinforcement learning,'' \emph{Machine Learning}, vol.~8, no.~3, pp. 229--256, 1992.

\bibitem{RL_book}
C.~Szepesv{\'a}ri, \emph{Algorithms for Reinforcement Learning}, 1st~ed.\hskip 1em plus 0.5em minus 0.4em\relax Springer, Cham, 2010.

\bibitem{ETD1}
H.~Yu, ``On convergence of emphatic temporal-difference learning,'' in \emph{Proc. Conference on Learning Theory}, vol.~40, 2015, pp. 1724--1751.

\bibitem{ETD2}
R.~S. Sutton, A.~R. Mahmood, and M.~White, ``An emphatic approach to the problem of off-policy temporal-difference learning,'' \emph{J. Mach. Learn. Res.}, vol.~17, no.~1, p. 2603–2631, 2016.

\bibitem{Sayed_diffusion}
J.~Chen and A.~H. Sayed, ``Diffusion adaptation strategies for distributed optimization and learning over networks,'' \emph{IEEE Transactions on Signal Processing}, vol.~60, no.~8, pp. 4289--4305, 2012.

\bibitem{Average1}
D.~Kempe, A.~Dobra, and J.~Gehrke, ``Gossip-based computation of aggregate information,'' in \emph{44th Annual IEEE Symposium on Foundations of Computer Science, 2003. Proceedings.}, 2003, pp. 482--491.

\bibitem{Average2}
J.~Liu and A.~S. Morse, ``Asynchronous distributed averaging using double linear iterations,'' in \emph{2012 American Control Conference (ACC)}, 2012, pp. 6620--6625.

\bibitem{JADBABAIE2012210}
A.~Jadbabaie, P.~Molavi, A.~Sandroni, and A.~Tahbaz-Salehi, ``Non-{Bayesian} social learning,'' \emph{Games Econ. Behav.}, vol.~76, no.~1, pp. 210--225, Sep 2012.

\bibitem{Nedic2017}
A.~Nedić, A.~Olshevsky, and C.~A. Uribe, ``Fast convergence rates for distributed non-bayesian learning,'' \emph{IEEE Trans. Autom. Control}, vol.~62, no.~11, pp. 5538--5553, Mar. 2017.

\bibitem{ASL}
V.~Bordignon, V.~Matta, and A.~H. Sayed, ``Adaptive social learning,'' \emph{IEEE Trans. Inf. Theory}, vol.~67, no.~9, pp. 6053--6081, Jul. 2021.

\bibitem{SL_book}
V.~Matta, V.~Bordignon, and A.~H. Sayed, \emph{Social Learning: Opinion Formation and Decision-Making over Graphs}.\hskip 1em plus 0.5em minus 0.4em\relax NOW Book Series on Information and Learning Sciences, 2024.

\bibitem{Malek}
M.~Khammassi, V.~Bordignon, V.~Matta, and A.~H. Sayed, ``Adaptive social learning for tracking rare transition {M}arkov chains,'' \emph{Proc. European Signal Processing Conference (EUSIPCO)}, pp. 1032--1036, 2024.

\bibitem{Malek2}
------, ``Fundamental social learning scaling law for tracking hidden {M}arkov models,'' in \emph{Proc. IEEE Int. Conf. Acoust., Speech Signal Process. (ICASSP)}, 2025, pp. 1--5.

\bibitem{Sutton}
R.~S. Sutton and A.~G. Barto, \emph{Reinforcement Learning: An Introduction}, 2nd~ed.\hskip 1em plus 0.5em minus 0.4em\relax MIT Press, 2018.

\bibitem{Boyd}
S.~Boyd, A.~Ghosh, B.~Prabhakar, and D.~Shah, ``Randomized gossip algorithms,'' \emph{IEEE Transactions on Information Theory}, vol.~52, no.~6, pp. 2508--2530, 2006.

\bibitem{pettingzoo}
J.~Terry, B.~Black, N.~Grammel, M.~Jayakumar, A.~Hari, R.~Sullivan, L.~S. Santos, C.~Dieffendahl, C.~Horsch, R.~Perez-Vicente, N.~L. Williams, and Y.~Lokesh, ``Pettingzoo: Gym for multi-agent reinforcement learning,'' \emph{Advances in Neural Information Processing Systems}, vol.~34, pp. 15\,032--15\,043, 2021.

\bibitem{Actor_critic2}
K.~Zhang, Z.~Yang, H.~Liu, T.~Zhang, and T.~Basar, ``Fully decentralized multi-agent reinforcement learning with networked agents,'' in \emph{Proc. International Conference on Machine Learning}, vol.~80, 2018, pp. 5872--5881.

\end{thebibliography}
\end{document}